\documentclass[11pt]{article}

\usepackage[utf8]{inputenc}
\usepackage[T1]{fontenc}
\usepackage{multicol}
\usepackage{amsthm}
\usepackage{amsmath}
\usepackage{amssymb}
\usepackage{mathtools}
\usepackage{dsfont}
\usepackage{bm}
\usepackage{bbm}
\usepackage{xparse}
\usepackage{physics}
\usepackage{empheq}
\usepackage{url}
\usepackage{hyperref}
\usepackage[affil-it]{authblk}
\usepackage{enumitem}
\usepackage{tikz}
\usetikzlibrary{quotes,angles,calc}
\usepackage{rotating}
\usepackage{graphicx}
\usepackage[linesnumbered,ruled,vlined]{algorithm2e}

\usepackage{pgfplots}
\pgfplotsset{compat = newest}

\usepackage[top=3cm,bottom=2cm,right=2cm,left=2cm]{geometry}

\theoremstyle{definition}
\newtheorem{theo}{Theorem}[section]
\newtheorem{lem}[theo]{Lemma}

\newtheorem{prop}[theo]{Proposition}
\newtheorem{defi}[theo]{Definition}

\newtheorem{theo*}{Theorem}

\theoremstyle{remark}
\newtheorem*{rk}{Remark}

\DeclareMathOperator{\Poi}{\text{Poi}}
\DeclareMathOperator{\Ber}{\text{Ber}}
\DeclareMathOperator{\Bin}{\text{Bin}}
\DeclareMathOperator{\maxi}{\text{maximize}}

\DeclareMathOperator{\st}{\text{subject to}}

\DeclarePairedDelimiterX\set[1]\lbrace\rbrace{#1}

\colorlet{darkgreen}{green!40!black}

\title{\bfseries Tight Approximation Guarantees for \\ Concave Coverage Problems}
\author{Siddharth Barman\footnote{Indian Institute of Science, Bangalore, India. \href{mailto:barman@iisc.ac.in}{\texttt{barman@iisc.ac.in}}} \qquad Omar Fawzi\footnote{Univ Lyon, ENS Lyon, UCBL, CNRS, Inria,  LIP, F-69342, Lyon Cedex 07, France. \href{mailto:omar.fawzi@ens-lyon.fr}{\texttt{omar.fawzi@ens-lyon.fr}}} \qquad Paul Fermé\footnote{Univ Lyon, ENS Lyon, UCBL, CNRS, Inria, LIP, F-69342, Lyon Cedex 07, France. \href{mailto:paul.ferme@ens-lyon.fr}{\texttt{paul.ferme@ens-lyon.fr}}}
}
\date{}

\begin{document}

\maketitle

\begin{abstract}
In the maximum coverage problem, we are given subsets $T_1, \ldots, T_m$ of a universe $[n]$ along with an integer $k$ and the objective is to find a subset $S \subseteq [m]$ of size $k$ that maximizes $C(S) := \abs{\bigcup_{i \in S} T_i}$. It is a classic result that the greedy algorithm for this problem achieves an optimal approximation ratio of $1-e^{-1}$.

In this work we consider a generalization of this problem wherein an element $a$ can contribute by an amount that depends on the number of times it is covered. Given a concave, nondecreasing function $\varphi$, we define $C^{\varphi}(S) \coloneqq \sum_{a \in [n]}w_a\varphi(\abs{S}_a)$, where $\abs{S}_a = \abs{\set{i \in S : a \in T_i}}$. The standard maximum coverage problem corresponds to taking $\varphi(j) = \min\{j,1\}$. For any such $\varphi$, we provide an efficient algorithm that achieves an approximation ratio equal to the \emph{Poisson concavity ratio} of $\varphi$, defined by $\alpha_{\varphi} := \min_{x \in \mathbb{N}^*} \frac{\mathbb{E}[\varphi(\Poi(x))]}{\varphi(\mathbb{E}[\Poi(x)])}$. Complementing this approximation guarantee, we establish a matching NP-hardness result when $\varphi$ grows in a sublinear way. 

As special cases, we improve the result of~\cite{BFGG20} about maximum multi-coverage, that was based on the unique games conjecture, and we recover the result of~\cite{DMMS20} on multi-winner approval-based voting for geometrically dominant rules. Our result goes beyond these special cases and we illustrate it with applications to distributed resource allocation problems, welfare maximization problems and approval-based voting for general rules. 
\end{abstract}

\section{Introduction}
Coverage functions are central objects of study in combinatorial optimization. Problems related to optimizing such functions arise in multiple fields, such as operations research~\cite{CFN77}, machine learning~\cite{FK14}, algorithmic game theory~\cite{DV15}, and information theory~\cite{BF17}. The most basic covering problem is the \emph{maximum coverage} one. In this problem, we are given subsets $T_1, \ldots, T_m$ of a universe $[n]$, along with a positive integer $k$, and the objective is to find a size-$k$ subset $S \subseteq [m]$ that maximizes the coverage function $C(S) \coloneqq  \abs{\bigcup_{i \in S} T_i}$. A fundamental result in the field of approximation algorithms establishes that an approximation ratio of $1-e^{-1}$ can be achieved for this problem in polynomial-time~\cite{Hochbaum96} and, in fact, this approximation guarantee is tight, under the assumption that ${\rm P} \not= {\rm NP}$~\cite{Feige98}.

Note that in the maximum coverage problem, an element $a \in [n]$ is counted at most once in the objective, even if $a$ appears in several selected sets. However, if we think of elements $a \in [n]$ as goods or resources, there are many settings wherein the utility indeed increases with the number of copies of $a$ that get accumulated. Motivated, in part, by such settings, we consider a generalization of the maximum coverage problem where an element $a$ can contribute by an amount that depends on the number of times it is covered. 

Given a function $\varphi: \mathbb{N} \rightarrow \mathbb{R}_+$, an integer $k \in \mathbb{N}$, a universe of elements $[n]$, positive weights $w_a$ for each $a \in [n]$, and  subsets $T_1,\ldots,T_m \subseteq [n]$, the $\varphi$-\textsc{MaxCoverage} problem entails maximizing $C^{\varphi}(S) \coloneqq \sum_{a \in [n]}w_a  \varphi(\abs{S}_a)$ over subsets $S \subseteq [m]$ of cardinality $k$; here $\abs{S}_a = \abs{\set{i \in S : a \in T_i}}$.

This work focuses on functions $\varphi$ that are nondecreasing and concave (i.e., $\varphi(i+2) - \varphi(i+1) \leq \varphi(i+1) - \varphi(i)$ for $i \in \mathbb{N}$). We will also assume that the function $\varphi$ is normalized in the sense that $\varphi(0) = 0$ and $\varphi(1) = 1$.\footnote{One can always replace a generic $\varphi$ to a normalized one without changing the optimal solutions through a simple affine transformation.} Our approximation guarantees are in terms of the \emph{Poisson concavity ratio} of $\varphi$, which we define as follows:
\begin{align}
\label{eq:def-alpha-varphi}
\alpha_{\varphi} := \inf_{x \in \mathbb{N}^*} \ \frac{\mathbb{E}[\varphi(\Poi(x))]}{\varphi(\mathbb{E}[\Poi(x)])} =  \inf_{x \in \mathbb{N}^*} \ \frac{\mathbb{E}[\varphi(\Poi(x))]}{\varphi(x)} \ .
\end{align} 

Here $\Poi(x)$ denotes a Poisson-distributed random variable with parameter $x$. We will write $\alpha_{\varphi}(x) \coloneqq \frac{\mathbb{E}[\varphi(\Poi(x))]}{\varphi(x)}$, with $\alpha_{\varphi}(0) = 1$, and hence (see Proposition \ref{prop:minNotInf}), $\alpha_{\varphi}  = \min_{x \in  \mathbb{N}^*} \alpha_{\varphi}(x) = \inf_{x \in \mathbb{R}_+} \alpha_{\varphi}(x)$.\footnote{We require $\varphi$ to be defined for nonnegative integers and will extend it over $\mathbb{R}_+$ by considering its piecewise linear extension.}

Our main result is that the $\varphi$-\textsc{MaxCoverage} problem admits an efficient $\alpha_\varphi$-approximation algorithm, when $\varphi$ is normalized nondecreasing concave, and this approximation guarantee is tight when $\varphi$ grows sublinearly. Formally,  

\begin{theo*}
For any normalized nondecreasing concave function $\varphi$, there exists a polynomial-time $\alpha_\varphi$-approximation algorithm for the $\varphi$-\textsc{MaxCoverage} problem. Furthermore, for $\varphi(n) = o(n)$, it is {\rm NP}-hard to approximate the $\varphi$-\textsc{MaxCoverage} problem within a factor better than $\alpha_\varphi + \varepsilon$, for any constant $\varepsilon >0$.
\label{theo:main}
\end{theo*}

Before detailing the proof of the theorem, we provide a few remarks and connections to related work.

\paragraph{Applications and related work}
We can directly reduce  the standard  maximum coverage problem to $\varphi$-\textsc{MaxCoverage} by setting $\varphi(j) = \min\{ j,1\}$. In this case $\alpha_{\varphi} = 1 - e^{-1}$. One can also encapsulate, within our framework, the $\ell$-\textsc{MultiCoverage} problem studied in~\cite{BFGG20} by instantiating $\varphi(j) = \min\{ j,\ell\}$. In this setting, we recover the approximation ratio $\alpha_{\varphi} = 1-\frac{\ell^{\ell}e^{-\ell}}{\ell!}$, which matches the approximation guarantee obtained in~\cite{BFGG20} (see Proposition \ref{prop:lCover}). Note that the hardness result in~\cite{BFGG20} was based on the Unique Games Conjecture, whereas the current work proves that this guarantee is tight under ${\rm P} \not= {\rm NP}$.

Another application of $\varphi$-\textsc{MaxCoverage} is in the context of multiwinner elections that entail selecting $k$ (out of $m$) candidates with the objective of maximizing the cumulative utility of $n$ voters; here, the utility of each voter $a \in [n]$ increases as more and more approved (by $a$) candidates get selected. One can reduce multiwinner elections to a coverage problem by considering subset $T_i \subseteq [n]$ as the set of voters that approve of candidate $i \in [m]$ and $\varphi(j)$ as the utility that an agent achieves from $j$ approved selections.\footnote{Indeed, for a subset of candidates $S \subseteq [m]$, the utility of a voter $a \in [n]$ is equal to $\varphi ( \abs{S}_a)$, with $\abs{S}_a = \abs{\set{i \in S : a \in T_i}}$.} Addressing multiwinner elections in this standard utilitarian model, Dudycz et al.~\cite{DMMS20} obtain tight approximation guarantees for some well-studied classes of utilities. Specifically, the result in \cite{DMMS20} applies to the classic \emph{proportional approval voting rule}, which assigns a utility of $\sum_{i=1}^j \frac{1}{i}$ for $j$ approved selections. This voting rule corresponds to the coverage problem with $\varphi(j) = \sum_{i=1}^j \frac{1}{i}$. Section~\ref{subsection:multiwinner-election} shows that Theorem \ref{theo:main} holds for all the settings considered in \cite{DMMS20} and, in fact, applies more generally. In particular, the voting version of $\ell$-\textsc{MultiCoverage} (studied in \cite{Skowron2016}) can be addressed by Theorem \ref{theo:main}, but not by the result in~\cite{DMMS20}. Such a separation also arises when one truncates the proportional approval voting rule to, say, $\ell$ candidates, i.e., upon setting $\varphi(j) = \sum_{i=1}^{\min \{ j, \ell\}} \frac{1}{i}$. Given that multiwinner elections model multiple real-world settings (e.g., committee selection \cite{Skowron2016} and parliamentary proceedings \cite{Brill2018}), instantiations of $\varphi$-\textsc{MaxCoverage} in such social-choice contexts substantiate the applicability of our algorithmic result.

Coverage functions arise in numerous resource-allocation settings, such as sensor allocation~\cite{marden2008distributed}, job scheduling, and plant location~\cite{CFN77}. The goal, broadly, in such setups is to select $k$ subsets of resources (out of $m$ pre-specified ones) such that the \emph{welfare} generated by the selected resources is maximized--each resource's contribution to the welfare increases with the number of times it is selected. This problem can be cast as $\varphi$-\textsc{MaxCoverage} by setting $n$ to be the number of resources, $\{T_i\}_{i\in [m]}$ as the given collection of subsets, and $\varphi(j)$ to be the welfare contribution of a resource when it is covered $j$ times.\footnote{Formally, to capture specific welfare-maximization problems in their entirety we have to a consider $\varphi$-\textsc{MaxCoverage} with a matroid constraint, and not just bound the number of selected subsets by $k$. Details pertaining to matroid constraints and the reduction appear in Section~\ref{sec:matroid} and~\ref{subsection:welfare-maximization}, respectively.}  Here, we mention a specific allocation problem to highlight the relevance of studying $\varphi$ beyond the standard coverage and $\ell$-coverage formulations (see Section \ref{subsection:vehicle-target} for details): in the \textsc{Vehicle-Target Assignment} problem~\cite{Murphey00,PM19} the resources are $n$ targets and covering a target $j$ times contributes $\varphi^p(j) = \frac{1-(1-p)^j}{p}$ to the welfare; here, $p \in (0,1)$ is a given parameter. Interestingly, we find that for this problem, the approximation ratio $\alpha_{\varphi}$ we obtain can outperform the \emph{price of anarchy} (PoA), which corresponds to the approximation ratio of any method whereing the agents selfishly maximize their utilities (see Section~\ref{subsection:vehicle-target} for further discussion of this point). By contrast, in the resource allocation problem with $\varphi(j) = \min\{j, \ell\}$, the price of anarchy is equal to $\alpha_{\varphi}$; see \cite{CPM19} for details. Another allocation problem studied in \cite{PM19} corresponds to $\varphi$-\textsc{MaxCoverage} with $\varphi(j) = j^d$, for a given parameter $d \in (0,1)$. We refer to this instantiation as the $d$-\textsc{Power} function. 

Theorem~\ref{theo:main} gives us a tight approximation bound of $\alpha_{\varphi}$ for all the above-mentioned applications of $\varphi$-\textsc{MaxCoverage}. The values of $\alpha_{\varphi}$ for these instantiations are listed in Table \ref{figComp}. 

\begin{table}[!h]
  \begin{center}
    \begin{tabular}{|l|l|l|l|}
      \hline
      $\varphi$-\textsc{MaxCoverage}  & $\varphi(j)$ & $\alpha_{\varphi}$ & Derivation \\
      \hline
      \textsc{MaxCoverage} & $\min \{ j,1\}$ & $1 - e^{-1}$ & Prop.~\ref{prop:lCover} \\
      $\ell$-\textsc{MultiCoverage} & $\min\{ j,\ell\}$ & $1-\frac{\ell^{\ell}e^{-\ell}}{\ell!}$ & Prop.~\ref{prop:lCover}\\
       \textsc{Proportional Approval Voting} & $\sum_{i=1}^j \frac{1}{i}$ & $\alpha_{\varphi}(1) \simeq 0.7965\ldots$ & Prop.~\ref{prop:geoDominant} \\
      \textsc{Proportional Approval Voting} capped at $3$ & $\sum_{i=1}^{\min\{j,3\}} \frac{1}{i}$ & $\alpha_{\varphi}(1) \simeq 0.7910\ldots$ & Prop.~\ref{prop:phiLinear} \\
      $p$-\textsc{Vehicle-Target Assignment} & $\frac{1-(1-p)^j}{p}$ & $\frac{1 - e^{-p}}{p}$ & Prop.~\ref{prop:VTA} \\
      $0.1$-\textsc{Vehicle-Target Assignment} & $\frac{1-(1-0.1)^j}{0.1}$ & $\frac{1 - e^{-0.1}}{0.1} \simeq 0.9516\ldots$ & Prop.~\ref{prop:VTA} \\
      $0.1$-\textsc{Vehicle-Target Assignment} capped at $5$ & $\frac{1-(1-0.1)^{\min\{j,5\}}}{0.1}$ & $\alpha_{\varphi}(5) \simeq 0.8470\ldots$ & Prop.~\ref{prop:phiLinear} \\
           $d$-\textsc{Power} & $j^d$ & $e^{-1}\sum_{k=1}^{+\infty}\frac{k^d}{k!}$ & Prop.~\ref{prop:dPower} \\
      \hline
    \end{tabular}
  \end{center}
  \caption{Tight approximation ratios for particular choices of $\varphi$ in the $\varphi$-\textsc{MaxCoverage} problem.}
  \label{figComp}
\end{table}

It is relevant to compare the approximation guarantee, $\alpha_\varphi$, obtained in the current work with the approximation ratio based on the notion of curvature of submodular functions. Note that if $\varphi$ is nondecreasing and concave, then $C^{\varphi}$ is submodular. One can show, via a direct calculation, that for such a submodular $C^{\varphi}$ the curvature (as defined in \cite{CC84}) is given by $c = 1 - (\varphi(m) - \varphi(m-1))$ for instances with at most $m$ cover sets; see Proposition \ref{prop:SubCurv}. Therefore, the algorithm of Sviridenko et al.~\cite{SVW17} provides an approximation ratio of $1 - c e^{-1}$ for the $\varphi$-\textsc{MaxCoverage} problem. We note that the {Poisson concavity ratio} $\alpha_{\varphi}$ is always greater than or equal to this curvature-dependent ratio (Proposition \ref{prop:BetterRatio}). Specifically, for $p$-\textsc{Vehicle-Target Assignment}, it is strictly better for all $p \notin \{0,1\}$ and for $\ell$-\textsc{MultiCoverage}, it is strictly better for all $\ell \geq 2$ as remarked in \cite{BFGG20}. Therefore, for the setting at hand, the current work improves the approximation guarantee obtained in \cite{SVW17}.

\paragraph{Remarks on the {Poisson concavity ratio} $\alpha_{\varphi}$.}
\label{para:rkpoi}
By Jensen's inequality along with the nonnegativity and concavity of $\varphi$, we have that $\alpha_{\varphi} \in [0,1]$.  We show that $\alpha_{\varphi}$ can be computed numerically up to any precision $\varepsilon > 0$, in time that is polynomial in $\frac{1}{\varepsilon}$. In fact, Proposition \ref{prop:ConvergenceRate} shows that $\alpha_{\varphi}(x) \geq 1 - \varepsilon$ for all $x \geq N_{\varepsilon} := \lceil \left(\frac{6}{\varepsilon}\right)^4 \rceil$. Thus, we can iterate over all $x \in \{1, 2, \ldots, N_\varepsilon \}$ and find $\min_{x \in [N_{\varepsilon}]} \alpha_{\varphi}(x)$ up to $\varepsilon$ precision (under reasonable assumptions on $\varphi$). This gives us a method to overall compute $\alpha_{\varphi}$, up to an absolute error of $2 \varepsilon$: if $\alpha_{\varphi} \leq 1 - \varepsilon$, then computing $\min_{x \in [N_{\varepsilon}]} \alpha_{\varphi}(x)$ (up to $\varepsilon$ precision) suffices. Otherwise, if $\alpha_{\varphi} \geq 1 - \varepsilon$, then $\alpha_\varphi(1) \leq 1$ provides the desired bound. 
Furthermore, we note that Proposition \ref{prop:minInt} shows that even if we consider $\alpha_{\varphi}(x)$ over all $x \in \mathbb{R}_+$, an infimum (i.e., the value of $\alpha_{\varphi}$) is achieved at an integer.

\paragraph{Further hardness under Gap-ETH}
Theorem \ref{theo:main} shows that, under the assumption ${\rm P} \not= {\rm NP}$, no polynomial-time algorithm can approximate $\varphi$-\textsc{MaxCoverage} within a better ratio than $\alpha_{\varphi}$ for sublinear $\varphi$. One natural question that arises is whether relaxing the running time constraint helps. More precisely, since there are $\binom{m}{k} = O(m^k)$ choices of $k$ cover sets among the $m$ available, a simple exhaustive search algorithm works in time $O(m^k)$. We can ask if FPT algorithms with respect to $k$, running in time $f(k) \cdot m^{o(k)}$ with $f$ an arbitrary function, can do better. As in \cite{DMMS20}, we use the result of \cite{Manurangsi20} to show in Theorem \ref{theo:GapHardness} that such algorithms cannot approximate $\varphi$-\textsc{MaxCoverage} within a better ratio than $\alpha_{\varphi}$ for sublinear $\varphi$, under the Gap-ETH hypothesis \cite{CCKLMNT17}; see Section \ref{section:GapETH} for more details. This means that the brute-force strategy is essentially the best, if one wants to get a better approximation ratio than $\alpha_{\varphi}$.

\paragraph{Proof techniques and organization}
In Section \ref{section:ApproxAlgo}, we present our approximation algorithm for the $\varphi$-\textsc{MaxCoverage}. The algorithm is an application of \emph{pipage rounding}, a technique introduced in \cite{AS04}, on a linear programming relaxation of $\varphi$-\textsc{MaxCoverage}. We show that the multilinear extension $F^{\varphi}$ of $C^{\varphi}$ is efficiently computable and thus, we can compute an integer solution $x^{\text{int}}$ from the optimal fractional one $x^*$ satisfying $C^{\varphi}(x^{\text{int}}) \geq F^{\varphi}(x^*)$. Using the notion of \emph{convex order} between distributions, we show that $F^{\varphi}(x^*) \geq \sum_{a \in [n]} w_a \mathbb{E}[\varphi(\Poi(|x^*|_a))]$, where $|x|_{a} = \sum_{i \in [m] : a \in T_i} x_i$. Comparing this to the value $\sum_{a \in [n]} w_a \varphi(|x^*|_a)$ taken by the linear program, we get a ratio given by the \emph{Poisson concavity ratio} $\alpha_{\varphi}$. The concavity of $\varphi$ is crucial at several steps of the proof: it guarantees that the natural relaxation can be written as a linear program, it is used to relate between sums of Bernouilli random variables and a Poisson random variable via the convex order, as well as for the fact that we can restrict the infimum in the definition of $\alpha_{\varphi}$ to integer values of $x$. The generalization to matroid constraints follows in a standard way and is presented in Section \ref{sec:matroid}.

In Section \ref{section:Hardness}, we present the hardness result for $\varphi$-\textsc{MaxCoverage}. For this, we define a generalization of the partitioning gadget of Feige \cite{Feige98}, extending also \cite{BFGG20}. Roughly speaking, for an integer $x_{\varphi} \in \mathbb{N}$, it is a collection of $x_{\varphi}$-covers  of the set $[n]$ (an $x$-cover is a collection of subsets such that each element $a \in [n]$ is covered $x$ times, or in other words, its $\varphi$-coverage is $\varphi(x)n$) that are incompatible in the sense that if we take an element from each one of these $x_{\varphi}$-covers, then the $\varphi$-coverage is bounded approximately by $\mathbb{E}[\varphi(\Poi(x_{\varphi}))] n$. Then, we construct an instance of $\varphi$-\textsc{MaxCoverage} from an instance of the {\rm NP}-hard problem \emph{Label Cover} (as in \cite{DMMS20}) using such a gadget with $x_{\varphi} \in \text{argmin}_{x \in \mathbb{N}} \alpha_{\varphi}(x)$. Having set up the partitioning gadget, the analysis of the reduction can be obtained by carefully generalizing the reductions of~\cite{BFGG20} and~\cite{DMMS20}.

In Section \ref{section:applications}, we present different domains of application of our result.

\section{Approximation Algorithm for $\varphi$-\textsc{MaxCoverage}}
\label{section:ApproxAlgo}

Fix a function $\varphi : \mathbb{N} \to \mathbb{R}_+$ that is normalized, nondecreasing and concave. The $\varphi$-\textsc{MaxCoverage} problem is defined as follows. The input to the problem is given by positive integers $n,m,t$ and $m$ subsets $T_{1}, \dots, T_{m}$ of the set $[n]$ (described as characteristic vectors), the weights $w_a \in \mathbb{Q}_+^*$ for $a \in [n]$ (described as a couple of bitstring of length $t$), as well as an integer $k \in \{1, \dots, m\}$.  The output is a subset $S \subseteq [m]$ of size $k$ that maximizes $C^{\varphi}(S) = \sum_{a \in [n]} w_a \varphi(|S|_a)$, where $|S|_{a} = |\{ i \in S : a \in T_i\}|$.

Note that the input to this problem can be specified using $n(m+2t) + O(\log nmt)$ bits. To reduce the number of parameters, we will assume that $t$ is polynomial in $n$ and $m$, so that a polynomial time algorithm for this problem means an algorithm that runs in time polynomial in $n$ and $m$. The counting function $\varphi$ is fixed and does not depend on the instance of the problem, but for a given instance the problem only depends on the values $\varphi(0), \varphi(1), \dots, \varphi(m)$. We assume that we have black box access to $\varphi$ and to ensure that all the algorithms run in polynomial time, we assume that $\varphi(j)$ can be described with a number of bits that is polynomial in $j$ and that this description can be computed in polynomial time.

We now describe the approximation algorithm for $\varphi$-\textsc{MaxCoverage} that we analyze. As described above, we follow the standard relax and round strategy, as in~\cite{BFGG20}. First, we define a natural convex relaxation.
\begin{defi}[Relaxed program]
  \begin{equation}
    \begin{aligned}
      &\maxi&& \sum_{a \in [n]} w_ac_a \\
      &\st&& c_a \leq \varphi(\abs{x}_a), \forall a \in [n],\text{ with }\abs{x}_a := \sum_{i \in [m] : a \in T_i} x_i\\
      &&& 0 \leq x_i \leq 1, \forall i \in [m]\\
      &&& \sum_{i=1}^m x_i = k \ .
    \end{aligned}
  \end{equation}
  \label{defi:relaxedProg}
\end{defi}

As previously mentioned, $\varphi$ is defined on $\mathbb{R}_+$ by extending it in a piecewise linear fashion on non-integral points. As such, the constraint $c_a \leq \varphi(\abs{x}_a)$ is equivalent to $m$ linear constraints. In fact, we can define $\varphi_j$ to be the linear function $\varphi_j(t) = (\varphi(j) - \varphi(j-1)) t  - (j-1) \varphi(j) + j \varphi(j-1)$ for $j \in [m]$. Since $\varphi$ is concave, we have that for all $t \in [0, m]$, $\varphi(t) = \min_{j \in [m]} \varphi_j(t)$. As such, the constraint $c_a \leq \varphi(\abs{x}_a)$ is equivalent to $c_a \leq \varphi_j(\abs{x}_a)$ for all $j\in [m]$ and so the program from Definition \ref{defi:relaxedProg} is a linear program. Overall there are $n+m$ variables and $(n+1)m + 1$ linear constraints, and by assumptions all the coefficients can be described using a number of bits that is polynomial in $n$ and $m$. Hence an optimal solution of this linear program can be found in polynomial time.

Also observe that the program from Definition \ref{defi:relaxedProg} is indeed a relaxation of the $\varphi$-\textsc{MaxCoverage} problem. To see this, given a set $S$ of size $k$, consider the characteristic vector $x \in \{0,1\}^m$ defined by $x_i = 1$ if and only if $i \in S$. Then for all $a \in [n]$, we can set $c_a = \varphi(\abs{x}_a) = \varphi(\abs{S}_a)$, and we get an objective value of $\sum_{a \in [n]}w_a\varphi(\abs{S}_a)$ which is exactly $C^{\varphi}(S)$. When solving the program from Definition \ref{defi:relaxedProg}, we get an optimal $x^* \in [0,1]^m$ which is in general not integral. Next, we describe a method to round it to an integral vector $x^{\text{int}} \in \set{0,1}^m$.

\paragraph{Rounding} 
For a submodular function $f : \{0,1\}^m \to \mathbb{R}$ , one can use pipage rounding~\cite{AS04, Vondrak07, CCPV11} to transform, in polynomial time, any fractional solution $x \in [0,1]^m$ satisfying $\sum_{i=1}^m x_i = k$ into an integral vector $x^{\text{int}} \in \set{0,1}^m$ such that $\sum_{i=1}^m x^{\text{int}}_i = k$ and $F(x^{\text{int}}) \geq F(x)$, where $F$ corresponds to the multilinear extension of $f$, provided that $F(x)$ is computable in polynomial time for a given $x$; see e.g.,~\cite[Lemma 3.4]{Vondrak07}. The multilinear extension $F : [0, 1]^m \rightarrow \mathbb{R}$ of $f$ is defined by $F(x_1,\ldots,x_m): = \mathbb{E}[f(X_1,\ldots,X_m)]$, where $X_i$ are independent random variables with $X_i \sim \Ber(x_i)$, i.e., $X_i \in \set{0,1}$ with $\mathbb{P}(X_i = 1) = x_i$. Note that $F(x) = f(x)$ for an integral vector $x \in \{0,1\}^m$.

We apply this strategy to $C^{\varphi}$, which is shown to be submodular in Proposition \ref{prop:SubCurv}, and the solution $x^*$ of the LP relaxation from Definition \ref{defi:relaxedProg}. Note that overall the algorithm is polynomial time, since here $F(x)$ is computable in polynomial time for a given $x$ (see Proposition \ref{prop:Fpoly}). We now analyze the value returned by the algorithm. Using the property of pipage rounding, with the notation $X = (X_1,\ldots,X_m)$ and $\Ber(x) = (\Ber(x_1),\ldots,\Ber(x_m))$, we get

\[
C^{\varphi}(x^{\text{int}}) = \mathbb{E}_{X \sim \Ber(x^{\text{int}})}[C^{\varphi}(X)] \geq \mathbb{E}_{X \sim \Ber(x^*)}[C^{\varphi}(X)] \ .
\]
Then it suffices to relate $\mathbb{E}_{X \sim \Ber(x^*)}[C^{\varphi}(X)]$ to the optimal value of the LP relaxation~\ref{defi:relaxedProg}, which can only be larger than the optimal value of the $\varphi$-\textsc{MaxCoverage} problem.

\begin{theo*}
  Let $x,c$ be a feasible solution of the program from Definition \ref{defi:relaxedProg} and $X \sim \Ber(x)$. Recalling the definition of $\alpha_{\varphi}$ and $\alpha_{\varphi}(j)$ from~\eqref{eq:def-alpha-varphi}, we have
  \[\mathbb{E}_{X \sim \Ber(x)}[C^{\varphi}(X)] \geq \left(\min_{j \in [m]} \alpha_{\varphi}(j)\right) \sum_{a \in [n]} w_ac_a\ .\]
  In particular, this implies that the described polynomial time algorithm has an approximation ratio of $\alpha_{\varphi}$:
  \[C^{\varphi}(x^{\text{int}}) \geq \alpha_{\varphi} \sum_{a \in [n]} w_ac^*_a \geq \alpha_{\varphi} \max_{S \subseteq [m] : \abs{S} = k} C^{\varphi}(S)\ .\]
  \label{theo:AlgoCard}
\end{theo*}

In order to prove this theorem, we need the following lemma:

\begin{lem}
    For $\varphi$ concave, and $p \in [0,1]^m$, we have:
    \[\mathbb{E}\Big[\varphi\Big(\sum_{i=1}^m\Ber(p_i)\Big)\Big] \geq \mathbb{E}\Big[\varphi\Big(\Poi\Big(\sum_{i=1}^m p_i\Big)\Big)\Big]\ .\]
  \label{lem:ConvexOrder}
\end{lem}

\begin{proof}
  The notion of \emph{convex order} discussed in \cite{StochasticOrders} allows us to prove this result. We say that $X \leq_{\text{cx}} Y \iff \mathbb{E}[f(X)] \leq \mathbb{E}[f(Y)]$ for any convex $f$. Thanks to Lemma 2.3 of \cite{BFGG20}, we have that for $p \in [0,1]$:
  \[\Ber(p) \leq_{\text{cx}} \Poi(p)\ .\]
  Since this order is preserved through convolution (Theorem 3.A.12 of \cite{StochasticOrders}), and the fact that $\sum_{i=1}^m \Poi(p_i) \sim \Poi\Big(\sum_{i=1}^m p_i\Big)$, we have:
  \[\sum_{i=1}^m\Ber(p_i) \leq_{\text{cx}}  \Poi\Big(\sum_{i=1}^m p_i\Big)\ .\]
  Applying this result to $-\varphi$, which is convex, concludes the proof.
\end{proof}

\begin{proof}[Proof of Theorem \ref{theo:AlgoCard}]
  By linearity of expectation and the fact that the weights $w_a$ are positive, it is sufficient to show that for all $a \in [n]$:
  \[ \mathbb{E}[C_a^{\varphi}(X)] \geq \left(\min_{j \in [m]} \alpha_{\varphi}(j) \right) c_a \ , \]
 where $C_a^{\varphi}(S) := \varphi(\abs{S}_a)$. Note that $\abs{X}_a = \sum_{i \in [m] : a \in T_i} X_i$, and thus:
  \begin{equation}
    \begin{aligned}
      \mathbb{E}[C_a^{\varphi}(X)] &=&& \mathbb{E}\Big[\varphi\Big(\sum_{i \in [m] : a \in T_i} X_i\Big)\Big] = \mathbb{E}\Big[\varphi\Big(\sum_{i \in [m] : a \in T_i}\Ber(x_i) \Big)\Big]\\
      &\geq&& \mathbb{E}\Big[\varphi\Big(\Poi\Big(\sum_{i \in [m] : a \in T_i} x_i\Big)\Big)\Big] \text{ thanks to Lemma \ref{lem:ConvexOrder}}\\
      &=&& \mathbb{E}[\varphi(\Poi(\abs{x}_a))] \geq \min \{ \alpha_{\varphi}(\lfloor \abs{x}_a \rfloor), \alpha_{\varphi}(\lceil \abs{x}_a \rceil) \} \varphi(\abs{x}_a) \quad \text{thanks to Proposition \ref{prop:minInt}}\\
      &\geq&& \left(\min_{j \in [m]} \alpha_{\varphi}(j) \right) \varphi(\abs{x}_a) \geq \left(\min_{j \in [m]} \alpha_{\varphi}(j) \right) c_a\ .
    \end{aligned}
  \end{equation}

\end{proof}

\subsection{Generalization to Matroid Constraints}
\label{sec:matroid}

Instead of taking a cardinality constraint $k$ on the size of the subset $S$, we look now at general matroid constraints on $S$. Specifically, as input, instead of $k$, we take a matroid $\mathcal{M}$ defined on $[m]$ and given by a set of linear constraints describing its base polytope $B(\mathcal{M})$. The output is a set $S \in \mathcal{M}$ that maximizes $C^{\varphi}(S)$. Note that the cardinality constraint considered above is the special case where $\mathcal{M}$ is the uniform matroid of all subsets of size at most $k$ and the base polytope $B(\mathcal{M}) = \{ x \in [0,1]^m : \sum_{i=1}^m x_i = k\}$.

We first note that in the order to establish Theorem \ref{theo:AlgoCard}, the cardinality constraint $\sum_{i=1}^m x_i = k$ is not used. Thus, since the pipage rounding strategy applies to matroid constraints $\mathcal{M}$ (see \cite[Lemma 3.4]{Vondrak07}), the strategy and the analysis of its efficiency generalize immediately when applied to the following linear program:

\begin{defi}[Relaxed program for matroid constraints]
  \begin{equation}
    \begin{aligned}
      &\maxi&& \sum_{a \in [n]} w_ac_a \\
      &\st&& c_a \leq \varphi(\abs{x}_a), \forall a \in [n]\\
      &&& 0 \leq x_i \leq 1, \forall i \in [m]\\
      &&& x \in B(\mathcal{M}) \quad \text{ the base polytope of } \mathcal{M} \ .
    \end{aligned}
  \end{equation}
  \label{defi:relaxedMatProg}
\end{defi}

\begin{theo*}
  Let $x,c$ a feasible solution of the program from Definition \ref{defi:relaxedMatProg} and $X \sim \Ber(x)$. Then:
  \[\mathbb{E}_{X \sim \Ber(x)}[C^{\varphi}(X)] \geq \left(\min_{j \in [m]} \alpha_{\varphi}(j)\right) \sum_{a \in [n]} w_ac_a \ .\]
  In particular, this implies that the described polynomial time algorithm has an approximation ratio of $\alpha_{\varphi}$:
  \[C^{\varphi}(x^{\text{int}}) \geq \alpha_{\varphi} \sum_{a \in [n]} w_ac^*_a \geq \alpha_{\varphi} \max_{S \in \mathcal{M}} C^{\varphi}(S)\ .\]
  \label{theo:AlgoMat}
\end{theo*}

\section{Hardness of Approximation for $\varphi$-\textsc{MaxCoverage}}
\label{section:Hardness}
In this section, we establish an inapproximability bound for the $\varphi$-\textsc{MaxCoverage} problem with weights $1$ under cardinality constraints. Throughout this section we use $\Gamma$ to denote the universe of elements and, hence, an instance of the $\varphi$-\textsc{MaxCoverage} problem consists of $\Gamma$, along with a collection of subsets $\mathcal{F} = \set{F_i \subseteq \Gamma}_{i=1}^m$  and an integer $k$. Recall that the objective of this problem is to find a size-$k$ subset $S \subseteq [m]$ that maximizes $C^{\varphi}(S) = \sum_{a \in \Gamma}\varphi(\abs{S}_a)$.

We establish the following theorem in this section:

\begin{theo*}
 It is NP-hard to approximate the $\varphi$-\textsc{MaxCoverage} problem for $\varphi(n) = o(n)$ within a factor greater that $\alpha_{\varphi} + \varepsilon$ for any $\varepsilon > 0$.
  \label{theo:Hardness}
\end{theo*}

Our reduction is based on a problem called $h$-\textsc{AryLabelCover}, which is equivalent to the more standard \textsc{GapLabelCover} problem as will be shown in Appendix~\ref{app:NPhardnessGap}.
\begin{defi}[$h$-\textsc{AryLabelCover}]
  An instance $\mathcal{G} = (V,E,[L],[R],\set{\pi_{e,v}}_{e \in E, v \in e})$ of $h$-\textsc{AryLabelCover} is characterized by an $h$-uniform regular hypergraph $(V, E)$ and constraints $\pi_{e,v} : [L] \rightarrow [R]$. Here, each $h$-uniform hyperedge represents a $h$-ary constraint. Additionally, for any labeling $\sigma : V \rightarrow [L]$, we have the following notions of strongly and weakly satisfied constraints:
  \begin{itemize}
  \item An edge $e = (v_1,\ldots,v_h) \in E$ is \emph{strongly satisfied} by $\sigma$ if:
    \[ \forall x,y \in [h], \pi_{e,v_x}(\sigma(v_x)) = \pi_{e,v_y}(\sigma(v_y))\ . \]
  \item An edge $e = (v_1,\ldots,v_h) \in E$ is \emph{weakly satisfied} by $\sigma$ if:
    \[ \exists x\not=y \in [h], \pi_{e,v_x}(\sigma(v_x)) = \pi_{e,v_y}(\sigma(v_y))\ . \]
  \end{itemize}
\end{defi}

\begin{prop}[$\delta,h$-\textsc{AryGapLabelCover}]
\label{prop:hardness-ary-lc}
  For any fixed integer $h \geq 2$ and fixed $\delta > 0$, there exists an $R_0$ such that for any integer $R \geq R_0$, it is {\rm NP}-hard for instances $\mathcal{G} = (V,E,[L],[R],\set{\pi_{e,v}}_{e \in E, v \in e})$ of $h$-\textsc{AryLabelCover} with right alphabet $[R]$ to distinguish between: 
  \begin{itemize}
  \item[\textbf{YES:}] There exists a labeling $\sigma$ that \emph{strongly satisfies} all the edges.
  \item[\textbf{NO:}] No labeling \emph{weakly satisfies} more than $\delta$ fraction of the edges.
  \end{itemize}
  \label{prop:AryGapLabelCover}
\end{prop}

\subsection{Partitioning System}

The key ingredient to prove Theorem~\ref{theo:Hardness} is a constant size combinatorial object called partitioning system, generalizing the work of Feige~\cite{Feige98} and~\cite{BFGG20}.
For any set $[n]$, $\mathcal{Q} \subseteq 2^{[n]}$, we overload the definition $C^{\varphi}(\mathcal{Q}) := \sum_{a \in [n]} \varphi(\abs{\mathcal{Q}}_a)$ with $\abs{\mathcal{Q}}_a:=\abs{\set{P \in \mathcal{Q} : a \in P}}$ and $C_a^{\varphi}(\mathcal{Q}) := \varphi(\abs{\mathcal{Q}}_a)$. Let us take $x_{\varphi} \in \text{argmin}_{x \in \mathbb{N}^*} \alpha_{\varphi}(x)$, thus $\alpha_{\varphi} = \alpha_{\varphi}(x_{\varphi})$.

We say that $\mathcal{Q}$ is an \emph{$x$-cover} of $x \in \mathbb{N}$ if every element of $[n]$ is covered $x$ times, so $C^{\varphi}(\mathcal{Q}) = n\varphi(x)$.

\begin{defi}
  An \emph{$([n],h,R,\varphi,\eta)$-partitioning system} consists of $R$ collections of subsets of $[n]$, $\mathcal{P}_1,\ldots,\mathcal{P}_R \subseteq 2^{[n]}$, that satisfy $\frac{x_{\varphi}n}{h} \in \mathbb{N}$, $x_{\varphi} \geq h$ and:
  \begin{enumerate}
  \item For every $i \in [R], \mathcal{P}_i$ is a collection of $h$ subsets $P_{i,1}, \ldots, P_{i,h} \subseteq [n]$ each of size $\frac{x_{\varphi}n}{h}$ which is an $x_{\varphi}$-cover.
  \item For any $T \subseteq [R]$ and $\mathcal{Q} = \set{P_{i,j(i)} : i \in T}$ for some function $j : T \rightarrow [h]$, we have $\abs{C^{\varphi}(\mathcal{Q}) -\psi^{\varphi}_{\abs{T},h} n} \leq \eta n$ where:
  \begin{align}
  \label{eq:def-phi}
   \psi^{\varphi}_{k,h} := \mathbb{E}\Big[\varphi\Big(\Bin\Big(k,\frac{x_{\varphi}}{h}\Big)\Big)\Big] \ .
   \end{align}
  \end{enumerate}
  \label{defi:PartSystem}
\end{defi}

\begin{rk}
  In particular, for any $\mathcal{Q} = \set{Q_1, \ldots, Q_k}$ with $Q_i$ of size $\frac{x_{\varphi}n}{h}$, we have that $C^{\varphi}(\mathcal{Q}) \leq n\varphi(k\frac{x_{\varphi}}{h})$. Indeed $C^{\varphi}(\mathcal{Q}) = \sum_{a \in [n]} \varphi(\abs{\mathcal{Q}}_a)$ with $\sum_{a \in [n]} \abs{\mathcal{Q}}_a = \sum_{i \in [k]} \abs{Q_i} = k \cdot \frac{x_{\varphi}n}{h}$. By concavity of $\varphi$ and Jensen's inequality, this function is maximized when all $\abs{\mathcal{Q}}_a$ are equals, where we get $n\varphi(k\frac{x_{\varphi}}{h})$.
\end{rk}
\begin{prop}
  For every choice of $R,h \in \mathbb{N}$ with $h \geq x_{\varphi}$, $\eta \in (0,1)$,  $n \geq \eta^{-2}R\varphi(R)^2\log(20(h+1))$ such that $\frac{x_{\varphi}n}{h} \in \mathbb{N}$, there exists an $([n],h,R,\varphi,\eta)$-partitioning system, which can be found in time exp($Rn\log(n))\cdot \text{poly}(h)$.
  \label{prop:Partitioning}
\end{prop}

The proof can be found in Appendix~\ref{app:Partitioning}.

\subsection{The Reduction}
  \begin{proof}[Proof of Theorem \ref{theo:Hardness}]
    Let $\varepsilon > 0$. Without loss of generality, we can assume that $\varepsilon < 1$. We show that it is NP-hard to reach an approximation greater than $\alpha_{\varphi} + \varepsilon$ for the $\varphi$-\textsc{MaxCoverage} problem, via a reduction from $\delta,h$-\textsc{AryGapLabelCover}.
  \begin{itemize}
    \item $\eta = \frac{\varphi(x_{\varphi})}{4x_{\varphi}} \varepsilon$, so $0 < \eta \leq \varepsilon < 1$,
    \item $h \geq x_{\varphi}$ such that $\abs{\psi^{\varphi}_{h,h} - \alpha_{\varphi}\varphi(x_{\varphi})} \leq \eta$ (see~\eqref{eq:def-phi} for the definition of $\psi^{\varphi}$); such a choice exists thanks to Proposition \ref{prop:UnboundBinPoi},
    \item $\theta$ such that for all $x \geq \theta$, $\frac{\varphi(x)}{x} \leq \eta$, which exists since $\varphi(x) = o(x)$,
    \item $\xi = \frac{x_{\varphi}}{\theta}$,
    \item $\delta = \frac{\eta}{2} \frac{\xi^3}{h^2}$,
    \item $R \geq h$ large enough for Proposition~\ref{prop:hardness-ary-lc} to hold.
\end{itemize}
    
Then, given an instance  $\mathcal{G} = (V,E,[L],[R], \Sigma, \set{\pi_{e,v}}_{e \in E, v \in e})$ of $\delta,h$-\textsc{AryGapLabelCover}, we construct an instance $(\Gamma, \mathcal{F}, k)$ of the $\varphi$-\textsc{MaxCoverage} problem with:

    \begin{itemize}
    \item $n$ a large enough integer to have the existence of $([n],h,R,\varphi,\eta)$-partitioning systems using Proposition \ref{prop:Partitioning}. Note that the size of these partitioning systems is independent of the size of the instance $\mathcal{G}$, and that one can find one of those in constant time, with relation to the size of the instance $\mathcal{G}$, thanks to Proposition \ref{prop:Partitioning}.
    \item $\Gamma = [n] \times E$,
    \item $k =\abs{V}$,
    \item Consider a  $([n],h,R,\varphi,\eta)$-partitioning system, and call $\mathcal{P} =\set{\mathcal{P}_1,\ldots,\mathcal{P}_R}$ the corresponding set of collections. Define sets $T_{\beta}^{e,v_j} = P_{\pi_{e,v_j}(\beta),j} \times \set{e}$ for $e = (v_1,\ldots,v_h) \in E, j \in [h], \beta \in [L]$. Then, choose as cover sets $F^v_{\beta} := \bigsqcup_{e \in E:v \in e} T^{e,v}_{\beta}$ and take $\mathcal{F} := \set{F^v_{\beta}, v \in V, \beta \in [L]}$.
  \end{itemize}

We will now prove that if we are in a YES instance, we have that there exists $\mathcal{T}$ of size $k$ such that $C^{\varphi}(\mathcal{T}) \geq \varphi(x_{\varphi})\abs{\Gamma}$ (completeness). Moreover, if we are in a NO instance, then we have that for all $\mathcal{T}$ of size $k = \abs{V}$, $C^{\varphi}(\mathcal{T}) \leq (\alpha_{\varphi} + \varepsilon)\varphi(x_{\varphi})\abs{\Gamma}$ (soundness). Establishing these two properties would conclude the proof. 
In fact, an algorithm for $\varphi$-\textsc{MaxCoverage} achieving a factor strictly greater than $\alpha_{\varphi} + \varepsilon$ would allow us to decide whether we have YES or a NO instance of the {\rm NP}-hard problem $\delta,h$-\textsc{AryGapLabelCover}.

In order to achieve this, let us define $C^{\varphi,e} := \sum_{a \in [n] \times \set{e}} C^{\varphi}_a$. In particular, $C^{\varphi} = \sum_{a \in  \Gamma} C^{\varphi}_a = \sum_{e \in E}C^{\varphi,e}$. For $\mathcal{T} \subseteq \mathcal{F}$, we define the relevant part of $\mathcal{T}$ on $e$ by:
\[\mathcal{T}_e := \set{T_{\beta}^{e,v} : v \in e, \beta \in [L], F^v_{\beta} \in \mathcal{T}} = \set{F^v_{\beta} \cap ([n] \times \set{e}), F^v_{\beta} \in \mathcal{T}}\ .\]
Note that $C^{\varphi,e}(\mathcal{T}) = C^{\varphi,e}(\mathcal{T}_e)$, and in particular $C^{\varphi}(\mathcal{T}) = \sum_{e \in E} C^{\varphi,e}(\mathcal{T}_e)$. 

\subsubsection{Completeness}
Suppose the given $h$-\textsc{AryLabelCover} instance $\mathcal{G}$ is a YES instance. Then, there exists a labeling $\sigma : V \mapsto [L]$ which strongly satisfies all edges. Consider the collection of $\abs{V}$ subsets $\mathcal{T} := \set{F_{\sigma(v)}^v : v \in V}$. Fix $e = (v_1,\ldots,v_h) \in E$. Since $e$ is strongly satisfied by $\sigma$, there exists $r \in [R]$ such that $\pi_{e,v_i}(\sigma(v_i)) = r$ for all $i \in [h]$. Thus, $\mathcal{T}_e = \set{T_{\sigma(v_i)}^{e,v_i}}_{i \in [h]} = \set{P_{r,i} \times \set{e}}_{i \in [h]}$ is an $x_{\varphi}$-cover of $[n] \times \set{e}$, and so $C^{\varphi,e}(\mathcal{T}_e) = n\varphi(x_{\varphi})$. Thus $C^{\varphi}(\mathcal{T}) = \sum_{e \in E}C^{\varphi,e}(\mathcal{T}_e) = \abs{E}\varphi(x_{\varphi})n = \varphi(x_{\varphi})\abs{\Gamma}$.

\subsubsection{Soundness}
Suppose the given $h$-\textsc{AryLabelCover} instance $\mathcal{G}$ is a NO instance. Let us prove the contrapositive of the soundness: we suppose that there exists $\mathcal{T}$ of size $k = \abs{V}$ such that  $C^{\varphi}(\mathcal{T}) > (\alpha_{\varphi} + \varepsilon)\varphi(x_{\varphi})\abs{\Gamma}$. Let us show that there exists a labeling $\sigma$ that weakly satisfies a strictly larger fraction of the edges than $\delta$.

For every vertex $v \in V$, we define $L(v) := \set{\beta \in [L] : F_{\beta}^v \in \mathcal{T}}$ to be the candidate set of labels that can be associated with the vertex $v$. We extend this definition to hyperedges $e = (v_1,\ldots,v_h)$ where we define $L(e) := \bigcup_{i \in [h]} L(v_i)$ to be the \emph{multiset} of all labels associated with the edge. Note that $\abs{\mathcal{T}_e}=\abs{L(e)}$.

We say that $e = (v_1,\ldots,v_h) \in E$ is \emph{consistent} if and only if $\exists x \not= y \in [h], \pi_{e,v_x}(L(v_x)) \cap \pi_{e,v_y}(L(v_y)) \not= \emptyset$. We then decompose $E$ in three parts:
\begin{itemize}
\item $B$ is the set of edges $e \in E$ with $\abs{L(e)} \geq \frac{h}{\xi}$.
\item $N$ is the set of consistent edges $e \in E$ with $\abs{L(e)} < \frac{h}{\xi}$.
\item $I = E - (B \cup N)$ is the set of inconsistent edges $e \in E$ with $\abs{L(e)} < \frac{h}{\xi}$.
\end{itemize}

We want to show that the contribution of $N$ is not too small, which we will use to construct a labeling weakly satisfying enough edges. This comes from the following lemmas:

\begin{lem}
  $\sum_{e \in E} \abs{L(e)} = \abs{E} h$
  \label{lem:labelBound}
\end{lem}

\begin{proof}  
  Recall that our $h$-uniform hypergraph is regular; call $d$ its regular degree. In particular, we have that $d\abs{V} = \abs{E}h$. Note also that $\sum_{v \in V} \abs{L(v)} = \abs{\mathcal{T}} = \abs{V}$. Thus:
  \begin{equation}
    \begin{aligned}
      \sum_{e \in E} \abs{L(e)} = \sum_{e \in E} \sum_{v \in V : v \in e} \abs{L(v)} = \sum_{v \in V} \sum_{e \in E: v \in e} \abs{L(v)} = d \sum_{v \in V} \abs{L(v)} = d \abs{V} = \abs{E}h\ .
    \end{aligned}
  \end{equation}
\end{proof}

Next, we bound the contribution of $B$:
\begin{lem}
  $\sum_{e \in B} C^{\varphi,e}(\mathcal{T}_e) \leq \frac{\varepsilon}{4}\varphi(x_{\varphi})\abs{\Gamma}$.
  \label{lem:contribB}
\end{lem}
\begin{proof}
  We have:
  \begin{equation}
    \begin{aligned}
      \sum_{e \in B} C^{\varphi,e}(\mathcal{T}_e) &\leq&& \sum_{e \in B} n\varphi\Big(\abs{L(e)}\frac{x_{\varphi}}{h}\Big) \quad \text{by the remark on Definition \ref{defi:PartSystem} and } \abs{\mathcal{T}_e} = \abs{L(e)}\\
      &\leq&& \abs{B} \cdot n\varphi\Big(\frac{\sum_{e \in B} \abs{L(e)}}{\abs{B}}\frac{x_{\varphi}}{h}\Big) \quad \text{by Jensen's inequality on concave } \varphi\\
      &\leq&&  \abs{B} \cdot n\varphi\Big(\frac{\abs{E}h}{\abs{B}}\frac{x_{\varphi}}{h}\Big) \quad \text{since } \varphi \text{ nondecreasing and } \sum_{e \in B} \abs{L(e)} \leq \abs{E}h \text{ by Lemma \ref{lem:labelBound}}\\
      &=&& \frac{\varphi\big(\frac{\abs{E}x_{\varphi}}{\abs{B}}\big)}{\frac{\abs{E}x_{\varphi}}{\abs{B}}} x_{\varphi} \abs{\Gamma} \ .
    \end{aligned}
  \end{equation}

We have seen that $\sum_{e \in B} \abs{L(e)} \leq \abs{E} h$, but $\sum_{e \in B} \abs{L(e)} \geq \abs{B} \frac{h}{\xi}$ by definition of $B$, so we have that $\frac{\abs{B}}{\abs{E}} \leq \xi$. Thus $\frac{\abs{E}x_{\varphi}}{\abs{B}} \geq \frac{x_{\varphi}}{\xi} = \theta$. By definition of $\theta$, we get that $\sum_{e \in B} C^{\varphi,e}(\mathcal{T}_e) \leq \eta x_{\varphi} \abs{\Gamma} = \frac{\varepsilon}{4} \varphi(x_{\varphi})\abs{\Gamma}$.
\end{proof}

In order to bound the contribution of $I$, we will prove a property on inconsistent edges:

\begin{prop}
 Let $e = (v_1,\ldots,v_h) \in E$ be an inconsistent hyperedge with respect to $\mathcal{T}$. Then we have that $\abs{C^{\varphi,e}(\mathcal{T}_e) - \psi^{\varphi}_{\abs{L(e)},h}n } \leq \eta n$.
  \label{prop:inconsistent}
\end{prop}

\begin{proof}
  Since $e$ is inconsistent, $\forall x \not= y \in [h], \pi_{e,v_x}(L(v_x)) \cap \pi_{e,v_y}(L(v_y)) = \emptyset$. Therefore, for every $i \in [R]$, there is at most one $v \in e$ such that $i \in \pi_{e,v}(L(v))$, i.e., $\mathcal{T}_e$ intersects with $\mathcal{P}_i \times \set{e}$ in at most one subset. This gives us a subset $T \subseteq [R]$ and a function $j : T \rightarrow [h]$ such that $\mathcal{T}_e = \set{P_{i,j(i)} \times \set{e} : i \in T}$. As a consequence, $\abs{T} = \abs{\mathcal{T}_e} = \abs{L(e)}$ and by the second condition of the partitioning system, we get the expected result.
\end{proof}

Now, we can bound the contribution of $I$:

\begin{lem}
  $\sum_{e \in I} C^{\varphi,e}(\mathcal{T}_e) \leq (\alpha_{\varphi} + \frac{\varepsilon}{2})\varphi(x_{\varphi})\abs{\Gamma}$.
  \label{lem:contribI}
\end{lem}

\begin{proof}
Thanks to Proposition \ref{prop:inconsistent}, we have:

  \begin{equation}
    \begin{aligned}
      \sum_{e \in I} C^{\varphi,e}(\mathcal{T}_e) \leq \sum_{e \in I} (\psi^{\varphi}_{\abs{L(e)},h} +\eta)n \leq \sum_{e \in E} (\psi^{\varphi}_{\abs{L(e)},h} +\eta)n \ ,
    \end{aligned} 
  \end{equation}
  since $I \subseteq E$ and $\psi^{\varphi}_{\abs{L(e)},h} \geq 0$. But $\sum_{e \in E} \abs{L(e)} = \abs{E}h$ by Lemma \ref{lem:labelBound} and $x \mapsto \psi^{\varphi}_{x,h}$ is concave thanks to Proposition \ref{prop:BinCon}, so we can use Jensen's inequality to get $\sum_{e \in E} \psi^{\varphi}_{\abs{L(e)},h} \leq \abs{E} \psi^{\varphi}_{\frac{\sum_{e \in E} \abs{L(e)}}{\abs{E}},h} = \abs{E}\psi^{\varphi}_{h,h}$ and thus:
  \begin{equation}
    \begin{aligned}
      \sum_{e \in I} C^{\varphi,e}(\mathcal{T}_e) \leq (\psi^{\varphi}_{h,h} +\eta)n\abs{E} \leq (\alpha_{\varphi}\varphi(x_{\varphi}) + 2\eta)\abs{\Gamma} \ ,
    \end{aligned} 
  \end{equation}
  by definition of $h$. This implies that the total contribution of inconsistent edges $I$ is at most $\sum_{e \in I} C^{\varphi,e}(\mathcal{T}_e) \leq (\alpha_{\varphi}\varphi(x_{\varphi}) + 2\eta)\abs{\Gamma} \leq (\alpha_{\varphi}+ \frac{\varepsilon}{2})\varphi(x_{\varphi})\abs{\Gamma}$ by definition of $\eta$.
\end{proof}

\begin{lem}
  $\frac{\abs{N}}{\abs{E}} \geq \xi\eta$.
  \label{lem:nice}
\end{lem}

\begin{proof}
Since we have supposed that $\sum_{e \in E} C^{\varphi,e}(\mathcal{T}_e) = C^{\varphi}(\mathcal{T}) > (\alpha_{\varphi} + \varepsilon)\varphi(x_{\varphi})\abs{\Gamma}$, and with the help of Lemmas \ref{lem:contribB} and \ref{lem:contribI}, we have that the contribution of $N$ is:

\[\sum_{e \in N}  C^{\varphi,e}(\mathcal{T}_e) > \frac{\varepsilon}{4}\varphi(x_{\varphi})\abs{\Gamma}\ .\]

However, we have that for $e \in N$ that $C^{\varphi,e}(\mathcal{T}_e) \leq  n\varphi\Big(\abs{\mathcal{T}_e}\frac{x_{\varphi}}{h}\Big) = n\varphi\Big(\abs{L(e)}\frac{x_{\varphi}}{h}\Big) \leq n \varphi\Big(\frac{x_{\varphi}}{\xi}\Big) \leq \frac{nx_{\varphi}}{\xi}$ thanks to the remark on Definition \ref{defi:PartSystem} and the bound $\abs{L(e)} < \frac{h}{\xi}$. This implies that:
\[\frac{\abs{N}}{\abs{E}} \geq \frac{\xi}{x_{\varphi}}\frac{\varepsilon \varphi(x_{\varphi})}{4} = \xi \eta \ .\]
\end{proof}

Finally, we construct a randomized labeling $\sigma : V \mapsto [L]$ as follows: for $v \in V$, if $L(v) \not= \emptyset$, set $\sigma(v)$ uniformly from $L(v)$, otherwise set it arbitrarily. We claim that in expectation, this labeling must weakly satisfy $\delta$ fraction of the hyperedges.

To see this, fix any $e = (v_1,\ldots,v_h) \in N$. Thus $\exists x \not= y \in [h], \pi_{e,v_x}(L(v_x)) \cap \pi_{e,v_y}(L(v_y)) \not= \emptyset$. Furthermore $\abs{L(v_x)},\abs{L(v_y)} \leq \frac{h}{\xi}$. Thus, we have that $\pi_{e,v_x}(L(v_x)) = \pi_{e,v_y}(L(v_y))$ with probability at least $\frac{1}{\abs{L(v_x)}\abs{L(v_y)}} \geq \Big(\frac{\xi}{h}\Big)^2$.

Therefore:
\begin{equation}
  \begin{aligned}
    &&& \mathbb{E}_{\sigma}\mathbb{E}_{e \sim E}[\sigma \text{ weakly satisfies } e]\\
    &\geq&& \xi \eta \mathbb{E}_{\sigma}\mathbb{E}_{e \sim E}[\sigma \text{ weakly satisfies } e | e \in N] \quad \text{by Lemma \ref{lem:nice}}\\
    &>&& \frac{\eta}{2} \frac{\xi^3}{h^2} = \delta \ .
  \end{aligned}
\end{equation}

In particular there exists some labeling $\sigma$ such that $\mathbb{E}_{e \sim E}[\sigma \text{ weakly satisfies } e] > \delta$, and thus the soundness is also proved.
\end{proof}

  \subsection{Further hardness under Gap-ETH}
  \label{section:GapETH}
  The Gap Exponential Time hypothesis states that, for some constant $\delta > 0$, there is no $2^{o(n)}$-time algorithm that, given $n$-variable $3$-SAT formula, can distinguish whether the formula is fully satisfiable or that it is not even $(1-\delta)$-satisfiable. Gap-ETH is a standard assumption in proving FPT hardness of approximation (see e.g. \cite{CCKLMNT17}). Under such hypothesis, Manurangsi showed the following theorem:

  \begin{theo}[\cite{Manurangsi20}, adapted to $(\delta,h)$-\textsc{AryGapLabelCover}]
    \label{theo:GapETH}
    Assuming Gap-ETH, for every $\delta > 0$, every $h \in \mathbb{N}, h \geq 2$ and any sufficiently large $R \in \mathbb{N}$ (depending on $\delta, h$), no $f(k) \cdot N^{o(k)}$-time algorithm can solve $(\delta,h)$-\textsc{AryGapLabelCover} with right alphabet $[R]$, where $k$ denotes the number of vertices in $h$-\textsc{AryLabelCover}, $N$ is the size of the instance, and $f$ can be any function.
  \end{theo}

  Such a statement can be made in terms of the $(\delta,h)$-\textsc{AryGapLabelCover} problem, since it can be shown to be equivalent to $\delta$-Gap-Label-Cover$(t,R)$ (see Appendix \ref{app:NPhardnessGap} for more details).

  Furthermore, in the previous reduction, the constructed instance $(\Gamma,k,\mathcal{F})$ sizes are $\abs{\Gamma} = n\abs{E}$ (with $n$ a constant independent of the size of the instance), $k=\abs{V}$, and $\abs{\mathcal{F}} = k \cdot L$. Therefore, plugin Theorem \ref{theo:GapETH} in the previous reduction leads to the following hardness result:

  \begin{theo*}
    \label{theo:GapHardness}
    Assuming Gap-ETH and $\varphi(n) = o(n)$, we cannot achieve an ($\alpha_{\varphi} + \varepsilon$)-approximation for the $\varphi$-\textsc{MaxCoverage} problem, even in $f(k) \cdot m^{o(k)}$-time, for any function $f$, with $m$ the number of cover sets and $k$ the cardinality constraint.
  \end{theo*}
  
\section{Applications}
\label{section:applications}
This section shows that instantiations of $\varphi$-\textsc{MaxCoverage} encapsulate and generalize multiple problems from fields such as computational social choice \cite{BCE16} and algorithmic game theory \cite{AGT}.

\subsection{Multiwinner Elections}
\label{subsection:multiwinner-election} 
 
As mentioned previously, multiwinner elections (with a utilitarian model for the voters) entail selection of $k$ (out of $m$) candidates that maximize the utility across $n$ voters. Here, the utility of each voter $a \in [n]$ increases with the number of approved (by $a$) selections. 
The work of Dudycz et al.~\cite{DMMS20} study the computational complexity of such elections and, in particular, address classic voting rules in which---for a specified sequence of nonnegative weights $(w_1, w_2, \ldots)$---voter $a$'s utility is equal to $\sum_{i=1}^j w_i$, when she approves of $j$ candidates among the selected ones. One can view this election exercise as a coverage problem by considering subset $T_i \subseteq [n]$ as the set of voters that approve of candidate $i \in [m]$ and $\varphi(j) = \sum_{i=1}^j w_i$. Indeed, for a subset of candidates $S \subseteq [m]$, the utility of a voter $a \in [n]$ is equal to $\varphi ( \abs{S}_a)$, with $\abs{S}_a = \abs{\set{i \in S : a \in T_i}}$. 

Dudycz et al.~\cite{DMMS20} show that if the weights satisfy $w_1 \geq w_2 \geq \ldots$ (i.e., bear a  diminishing returns property) along with geometric dominance ($w_{i } \cdot w_{i+2} \geq w_{i+1}^2$ for all $i \in \mathbb{N}^*$) and $\lim_{i \rightarrow \infty} w_i = 0$, then a tight approximation guarantee can be obtained for the election problem at hand. Note that the diminishing returns property implies that $\varphi(j) = \sum_{i=1}^j w_i$ is concave and $\lim_{i \rightarrow \infty} w_i = 0$ ensures that $\varphi$ is sublinear (see Proposition \ref{prop:thieleEqLim}). Hence, Theorem \ref{theo:main}, together with Proposition \ref{prop:geoDominant}, can be invoked to recover the result in \cite{DMMS20} where we get $\alpha_{\varphi} = \alpha_{\varphi}(1)$. In fact, Theorem \ref{theo:main} does not require geometric dominance among the weights and, hence, applies to a broader class of voting rules. For instance, the geometric dominance property does not hold if one considers the voting weights induced by $\ell$-\textsc{MultiCoverage}, i.e., $w_i = 1$, for $1 \leq i \leq \ell$, and $w_j = 0$ for $j > \ell$.  However, using Theorem \ref{theo:main}, we get that for this voting rule we can approximate the optimal utility within a factor of $\alpha_\varphi = 1-\frac{\ell^{\ell}e^{-\ell}}{\ell!}$ (see Proposition~\ref{prop:lCover}). Another example of such a separation arises if one truncates the proportional approval voting. The standard proportional approval voting corresponds to $w_i = \frac{1}{i}$, for all $i \in \mathbb{N}$ (equivalently,  $\varphi(j) = \sum_{i=1}^j \frac{1}{i}$) and falls within the purview of \cite{DMMS20}. While the truncated version with $\varphi(j) = \sum_{i=1}^{\min\{j, \ell\}} \frac{1}{i}$, for a given threshold $\ell$, does not satisfy geometric dominance, Theorem \ref{theo:main} continues to hold and provide a tight approximation ratio that can be computed numerically (see Proposition~\ref{prop:phiLinear} and Table \ref{figComp} for examples).

\subsection{Resource Allocation in Multiagent Systems}
\label{subsection:welfare-maximization}
A significant body of prior work in algorithmic game theory has addressed game-theoretic aspects of maximizing welfare among multiple (strategic) agents; see, e.g.,~\cite{PM19}. Complementing such results, this section shows that the optimization problem underlying multiple welfare-maximization games can be expressed in terms of $\varphi$-\textsc{MaxCoverage}. 

Specifically, consider a setting with $n$ resources, $k$ agents, and a (counting) function $\varphi: \mathbb{N} \mapsto \mathbb{R}_+$. Every agent $i$ is endowed with a collection of resource subsets $\mathcal{A}_i = \set{T^i_1, \ldots, T^i_{m_i}}  \subseteq 2^{[n]}$ (i.e., each $T^i_j \subseteq [n]$). The objective is to select a subset $A_i \in \mathcal{A}_i$, for all $i \in [k]$, so as to maximize $W^{\varphi}(A_1, A_2, \ldots, A_k)  \coloneqq \sum_{a \in [n]} w_a \ \varphi(\abs{A}_a)$. Here, $w_a \in \mathbb{R}_+$ is a weight associated with $a \in [n]$ and $\abs{A}_a \coloneqq  \abs{\set{i \in [k] : a \in A_i}}$. We will refer to this problem as the $\varphi$-\textsc{Resource Allocation} problem.

While $\varphi$-\textsc{Resource Allocation} does not directly reduce to $\varphi$-\textsc{MaxCoverage}, the next theorem shows that it corresponds to maximizing $\varphi$-coverage functions subject to a matroid constraint. Hence, invoking our result from Section~\ref{sec:matroid}, we obtain a tight $\alpha_\varphi$-approximation for $\varphi$-\textsc{Resource Allocation} (see Appendix \ref{app:HardnessMA} for the proof):

\begin{theo*}
For any normalized nondecreasing concave function $\varphi$, there exists a polynomial-time $\alpha_\varphi$-approximation algorithm for $\varphi$-\textsc{Resource Allocation}. Furthermore, for $\varphi(n) = o(n)$, it is {\rm NP}-hard to approximate $\varphi$-\textsc{Resource Allocation} within a factor better than $\alpha_\varphi + \varepsilon$, for any constant $\varepsilon >0$.
  \label{theo:HardnessMA}
\end{theo*}

\subsection{Vehicle-Target Assignment}
\label{subsection:vehicle-target}
\textsc{Vehicle-Target Assignment}~\cite{Murphey00,PM19} is another problem which highlights the applicability of coverage problems, with a concave $\varphi$. In particular, \textsc{Vehicle-Target Assignment} can be directly expressed as $\varphi$-\textsc{Resource Allocation}: the $[n]$ resources correspond to targets, the agents correspond to vehicles $i \in [k]$, each  with a collection of covering choices $\mathcal{A}_i \subseteq 2^{[n]}$, and $\varphi^p(j) = \frac{1-(1-p)^j}{p}$, for a given parameter $p \in (0,1)$. As limit cases, we define $\varphi^0(j) := \lim_{p \rightarrow 0} \varphi^p(j) = j$ and $\varphi^1(j) := 1$. Since $\varphi^p(j)$ is concave, by Proposition \ref{prop:VTA} and Theorem \ref{theo:HardnessMA}, we obtain a novel tight approximation ratio of $\alpha_{\varphi^p} = \frac{1 - e^{-p}}{p}$ for this problem. Also, one can look at the capped version of this problem, $\varphi^p_{\ell}(j) := \varphi^p(\min\{j,\ell\})$. In particular, we recover the $\ell$-\textsc{MultiCoverage} function when $p=0$. In Figure \ref{fig:VTATrunc}, we have plotted several cases of the tight approximations $\alpha_{\varphi^p_{\ell}}$ in function of $\ell$ for several values of $\ell$:

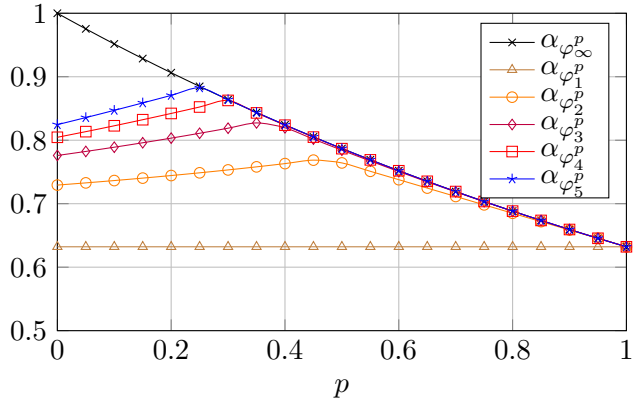
\begin{figure}[!h]
  \begin{center}
    \begin{tikzpicture}
      \begin{axis}[
          xmin = 0, xmax = 1,
          ymin = 0.5, ymax = 1,
          xtick distance = 0.2,
          ytick distance = 0.1,
          grid = both,
          width = 0.52\textwidth,
          height = 0.33\textwidth,
          legend cell align = {left},
          legend pos = north east,
          xlabel=$p$,
        ]
        \addplot[
          domain = 0:1,
          samples = 21,
          smooth,
          black,
          mark = x,
        ] table[x=p,y=alpha,col sep=comma] {VTATruncPlot.csv};
        \addplot[
          domain = 0:1,
          samples = 21,
          smooth,
          brown,
          mark = triangle,
        ] table[x=p,y=m1,col sep=comma] {VTATruncPlot.csv};
        \addplot[
          domain = 0:1,
          samples = 21,
          smooth,
          orange,
          mark = o,
        ] table[x=p,y=m2,col sep=comma] {VTATruncPlot.csv};
        \addplot[
          domain = 0:1,
          samples = 21,
          smooth,
          purple,
          mark = diamond,
        ] table[x=p,y=m3,col sep=comma] {VTATruncPlot.csv};
        \addplot[
          domain = 0:1,
          samples = 21,
          smooth,
          red,
          mark = square,
        ] table[x=p,y=m4,col sep=comma] {VTATruncPlot.csv};
        \addplot[
          domain = 0:1,
          samples = 21,
          smooth,
          blue,
          mark = star,
        ] table[x=p,y=m5,col sep=comma] {VTATruncPlot.csv};

        \legend{$\alpha_{\varphi^p_{\infty}}$, $\alpha_{\varphi^p_{1}}$, $\alpha_{\varphi^p_{2}}$, $\alpha_{\varphi^p_{3}}$, $\alpha_{\varphi^p_{4}}$, $\alpha_{\varphi^p_{5}}$ }
      \end{axis}
    \end{tikzpicture}
  \end{center}
  \caption{Tight approximation ratios $\alpha_{\varphi^p_{\ell}}$, where $\ell$ is the rank of the capped version of the $p$-\textsc{Vehicle-Target Assignment} problem. When $p=0$, we recover the $\ell$-coverage problem.}
  \label{fig:VTATrunc}
\end{figure}

Paccagnan and Marden~\cite{PM19} study the game-theoretic aspects of \textsc{Vehicle-target assignment}. A key goal in~\cite{PM19} is to bound the welfare loss incurred due to strategic selection by the $k$ vehicles, i.e., the selection of each $A_i \in \mathcal{A}_i$ by a self-interested vehicle/agent  $i \in [k]$. The loss is quantified in terms of the \emph{Price of Anarchy} (PoA). Formally, this performance metric is defined as ratio between the welfare of the worst-possible equilibria and the optimal welfare. Paccagnan and Marden~\cite{PM19} show that, for computationally tractable equilibrium concepts (in particular, for coarse correlated equilibria), tight price of anarchy bounds can be obtained via linear programs.

Note that our hardness result (Theorem~\ref{theo:main}) provides upper bounds on PoA of tractable equilibrium concepts--this follows from the observation that computing an equilibrium provides a specific method for finding a coverage solution. In \cite{CPM19} and in the particular case of the $\ell$-\textsc{MultiCoverage} problem, it is shown that this in fact an equality, i.e., PoA $=\alpha_{\varphi}$ if $\varphi(j) = \min\{j,\ell\}$ for all values of $\ell$. However, numerically comparing the approximation ratio for \textsc{Vehicle-Target Assignment}, $\alpha_{\varphi^p} = \frac{1 - e^{-p}}{p}$, with the optimal PoA bound, we note that $\alpha_{\varphi^p}$ can in fact be strictly greater than the PoA guarantee; see Figure~\ref{fig:VTA}.

Another form of the current problem, considered in \cite{PM19}, corresponds to $\varphi^d(j) = j^d$, for a given parameter $d \in (0,1)$. We refer to this instantiation as the $d$-\textsc{Power} function and for it obtain the approximation ratio $\alpha_{\varphi^d} = e^{-1}\sum_{k=1}^{+\infty}\frac{k^d}{k!}$ (Proposition \ref{prop:dPower}). In this case, the question whether the inequality PoA $\leq \alpha_{\varphi}$ is tight remains open; see Figure~\ref{fig:dPower}.

\begin{figure}[!h]
      \begin{center}
       \begin{tikzpicture}
          \begin{axis}[
            xmin = 0, xmax = 1,
            ymin = 0.5, ymax = 1,
            xtick distance = 0.2,
            ytick distance = 0.1,
            grid = both,
            width = 0.52\textwidth,
            height = 0.33\textwidth,
            legend cell align = {left},
            legend pos = north east,
            xlabel=$p$,
            ]
            \addplot[
              domain = 0:1,
              samples = 21,
              smooth,
              black,
              mark = x,
            ] table[x=p,y=alpha,col sep=comma] {VTAPoA.csv};
            \addplot[
              domain = 0:1,
              samples = 21,
              smooth,
              blue,
              mark = o,
            ] table[x=p,y=PoA,col sep=comma] {VTAPoA.csv};
            \addplot[
              domain = 0:1,
              samples = 21,
              smooth,
              red,
              mark = diamond,
            ] table[x=p,y=App,col sep=comma] {VTAPoA.csv};
            \legend{$\alpha_{\varphi^p} = \frac{1 - e^{-p}}{p}$, PoA$^{20}$,Curv $=1 - \frac{c}{e}$}
          \end{axis}
       \end{tikzpicture}
    \end{center}
    \caption{Comparison between the PoA and $\alpha_{\varphi}$ for the \textsc{Vehicle-Target Assignment} problem. Using the linear program found in \cite{PM19}, we were able to compute the blue curve PoA$^{20}$, the \emph{Price of Anarchy} of this problem for $m=20$ players. Since the PoA only decreases when the number of players grows, this means that PoA $< \alpha_{\varphi}$ in that case. As a comparison, the red curve Curv depicts the general approximation ratio (see \cite{SVW17}) obtained for submodular function with curvature $c$, with $c=1-\varphi^p(m) + \varphi^p(m-1)$ here.}
    \label{fig:VTA}
\end{figure}
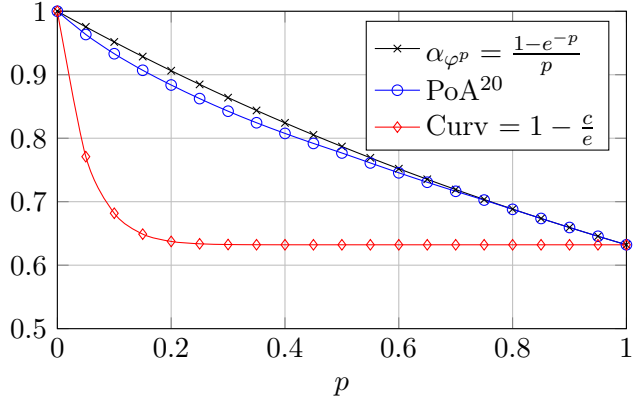

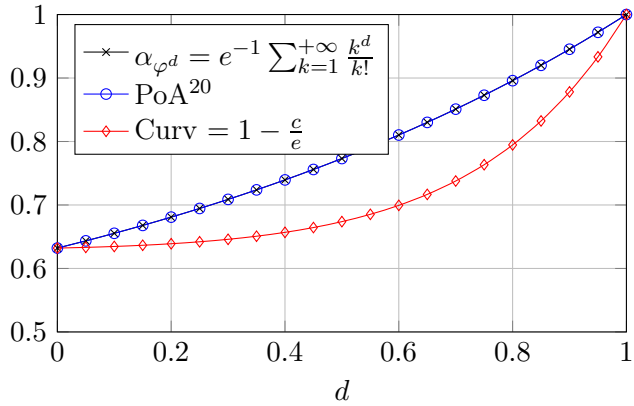
\begin{figure}[!h]
  \begin{center}
    \begin{tikzpicture}
      \begin{axis}[
          xmin = 0, xmax = 1,
          ymin = 0.5, ymax = 1,
          xtick distance = 0.2,
          ytick distance = 0.1,
          grid = both,
          width = 0.52\textwidth,
          height = 0.33\textwidth,
          legend cell align = {left},
          legend pos = north west,
          xlabel=$d$,
        ]
        \addplot[
          domain = 0:1,
          samples = 21,
          smooth,
          black,
          mark = x,
        ] table[x=d,y=alpha,col sep=comma] {dPowerPoA.csv};
        \addplot[
          domain = 0:1,
          samples = 21,
          smooth,
          blue,
          mark = o,
        ] table[x=d,y=PoA,col sep=comma] {dPowerPoA.csv};
        \addplot[
          domain = 0:1,
          samples = 21,
          smooth,
          red,
          mark = diamond,
        ] table[x=d,y=App,col sep=comma] {dPowerPoA.csv};
        \legend{$\alpha_{\varphi^d} = e^{-1}\sum_{k=1}^{+\infty}\frac{k^d}{k!}$, PoA$^{20}$, Curv $=1 - \frac{c}{e}$}
      \end{axis}
    \end{tikzpicture}
  \end{center}
  \caption{Comparison between the PoA and $\alpha_{\varphi}$ for the $d$-\textsc{Power} problem. Using the linear program found in \cite{PM19}, we were able to compute the blue curve PoA$^{20}$, the \emph{Price of Anarchy} of this problem for $m=20$ players. Here, the question whether the inequality PoA $\leq \alpha_{\varphi}$ is tight remains open. As a comparison, the red curve Curv depicts the general approximation ratio (see \cite{SVW17}) obtained for submodular function with curvature $c$, with $c=1-\varphi^d(m) + \varphi^d(m-1)$ here.}
  \label{fig:dPower}
\end{figure}

\subsection{Welfare Maximization for $\varphi$-Coverage}
Maximizing (social) welfare by partitioning items among agents is a key problem in algorithmic game theory; see, e.g., the extensive work on combinatorial auctions~\cite{AGT}. The goal here is to partition $t$ items among a set of $k$ agents such that the sum of values achieved by the agents---referred to as the social welfare---is maximized. That is, one needs to partition $[t]$ into $k$ pairwise disjoint subsets $A_1, A_2, \ldots, A_k$ with the objective of maximizing $\sum_{i=1}^k v_i(A_i)$. Here, $v_i(S)$ denotes the valuation that agent $i$ has for a subset of items $S \subseteq [t]$. 

When each agent's valuation $v_i$ is submodular, a tight $(1- e^{-1})$-approximation ratio is known for social welfare maximization~\cite{Vondrak07}. This section shows that improved approximation guarantees can be achieved if, in particular, the agents' valuations are $\varphi$-coverage functions. Towards a stylized application of such valuations, consider a setting in which each ``item'' $b \in [t]$ represents a bundle (subset) of goods $T_b \subseteq [n]$ and the value of an agent increases with the number of copies of any good $a \in [n]$ that get accumulated. Indeed, if each agent's value for $j$ copies of a good is $\varphi(j)$, then we have a $\varphi$-coverage function and the overall optimization problem is find a $k$-partition, $A_1, A_2, \ldots, A_k$, of $[t]$ that maximizes $\sum_{i=1}^k  \left( \sum_{a \in [n]} \varphi \left(\abs{A_i}_a \right) \right)$, where $ \abs{A_i}_a \coloneqq \{ b \in A_i : a \in T_b \}$. 

In the current setup, one can obtain an $\alpha_\varphi$ approximation ratio for social-welfare maximization by reducing this problem to $\varphi$-coverage with a matroid constraint, and applying the result from Section~\ref{sec:matroid}. Specifically, we can consider a partition matroid over the universe $[t] \times [k]$: for a bundle/item $b \in [t]$ and an agent $i \in [k]$, the element $(b,i)$ in the universe represents that bundle $b$ is assigned to agent $i$, i.e., $b \in A_i$. The partition-matroid constraint is imposed to ensure that each bundle $b$ is assigned to at most one agent. Furthermore, we can create $k$ copies of the underlying set of goods $[n]$ and set $T_{(b,i)} \coloneqq \{ (a, i) : a \in T_b \}$ to map the $\varphi$-coverage over the universe to the social-welfare objective. This, overall, gives us the desired  $\alpha_\varphi$ approximation guarantee.

\section*{Conclusion}
We have introduced the $\varphi$-\textsc{MaxCoverage} problem where having $c$ copies of element $a$ gives a value $\varphi(c)$. We have shown that when $\varphi$ is normalized, nondecreasing and concave, we can obtain an approximation guarantee given by the \emph{Poisson concavity ratio} $\alpha_{\varphi} := \min_{x \in \mathbb{N}} \frac{\mathbb{E}[\varphi(\Poi(x))]}{\varphi(\mathbb{E}[\Poi(x)])}$ and we showed it is tight for sublinear functions $\varphi$. The Poisson concavity ratio strictly beats the bound one gets when using the notion of curvature submodular functions, except in very special cases such as \textsc{MaxCoverage} where the two bounds are equal.

An interesting open question is whether there exists combinatorial algorithms that achieve this approximation ratio. As mentioned in \cite{BFGG20}, for the $\ell$-\textsc{MultiCoverage} with $\ell \geq 2$, which is the special case where $\varphi(x) = \min\{x,\ell\}$, the simple greedy algorithm only gives a $1 - e^{-1}$ approximation ratio, which is strictly less than the ratio $\alpha_{\varphi} = 1-\frac{\ell^{\ell}e^{-\ell}}{\ell!}$ in that case. Also, for any geometrically dominant vector $w=(\varphi(i+1)-\varphi(i))_{i \in \mathbb{N}}$ which is not $p$-geometric, such as \textsc{Proportional Approval Voting}, the greedy algorithm achieves an approximation ratio which is strictly less than $\alpha_{\varphi}$ (see Theorem 18 of \cite{DMMS20}).

Another open question is whether the hardness result remains true even when $\varphi(n) \not= o(n)$. A good example is given by $\varphi(0)=0$ and $\varphi(1+t) = 1 + (1-c)t$ with $c \in (0,1)$. We know that the problem is hard for $c=1$ but easy for $c=0$. One can show that the approximation ratio achieved by our algorithm is $\alpha_{\varphi} = 1 - \frac{c}{e}$ in that case (which is the same approximation ratio obtained from the curvature in \cite{SVW17}), but the tightness of this approximation ratio remains open.

\section*{Acknowledgements}

This research is supported by the French ANR project ANR-18-CE47-0011 (ACOM). SB gratefully acknowledges the support of a Ramanujan Fellowship (SERB - {SB/S2/RJN-128/2015}) and a Pratiksha Trust Young Investigator Award.

\bibliographystyle{plainurl}
\bibliography{MaxCoverage}

\newpage
\appendix

\section{General properties}
In this section, we will assume that $\varphi$ is specified over the nonnegative integers (i.e., $\varphi : \mathbb{N} \rightarrow \mathbb{R}_+$) and it is nondecreasing, concave, and normalized: $\varphi(0)=0$ and $\varphi(1)=1$. We will consider its piecewise linear extension on  $\mathbb{R}_+$ by defining $\varphi(x) := \lambda \varphi(\lfloor x \rfloor) + (1-\lambda)\varphi(\lceil x \rceil)$; here, parameter $\lambda\in [0,1]$ satisfies $x = \lambda\lfloor x \rfloor + (1-\lambda)\lceil x \rceil$. Note that the piecewise linear extension is also nondecreasing and concave.

\begin{prop}
For all $x \in \mathbb{R}_+$, we have $\alpha_{\varphi}(x) \geq \min\{\alpha_{\varphi}(\lfloor x \rfloor), \alpha_{\varphi}(\lceil x \rceil)\}$; here, $\alpha_{\varphi}(0) := \underset{x \rightarrow 0}{\lim} \alpha_{\varphi}(x) = 1$.
  \label{prop:minInt}
\end{prop}
\begin{proof}
For any $x \geq 1$, consider parameter $\lambda \in [0,1]$ such that $x = \lambda \lfloor x \rfloor + (1-\lambda) \lceil x \rceil$. Since $x \mapsto \mathbb{E}[\varphi(\Poi(x))]$ is concave (Proposition \ref{prop:PoiCon}), the following bound holds for all $x \geq 1$:
  \begin{align*}
      \mathbb{E}[\varphi(\Poi(x))] &\geq \lambda \mathbb{E}[\varphi(\Poi(\lfloor x \rfloor))] + (1-\lambda) \mathbb{E}[\varphi(\Poi(\lceil x \rceil))] \\      
      &=  \lambda \alpha_{\varphi}(\lfloor x \rfloor)\varphi(\lfloor x \rfloor) + (1-\lambda) \alpha_{\varphi}(\lceil x \rceil)\varphi(\lceil x \rceil) \quad \text{by definition of $\alpha_{\varphi}(x)$}\\
      &\geq \min\{\alpha_{\varphi}(\lfloor x \rfloor), \alpha_{\varphi}(\lceil x \rceil)\} \left(\lambda \varphi(\lfloor x \rfloor) + (1-\lambda) \varphi(\lceil x \rceil) \right)\\
      &=  \min\{\alpha_{\varphi}(\lfloor x \rfloor), \alpha_{\varphi}(\lceil x \rceil)\} \varphi(x) \quad \text{since $\varphi$ linear between integer points.} 
  \end{align*}
Therefore, $\alpha_{\varphi}(x) = \frac{\mathbb{E}[\varphi(\Poi(x))]}{\varphi(x)} \geq \min\{\alpha_{\varphi}(\lfloor x \rfloor), \alpha_{\varphi}(\lceil x \rceil)\}$.

Next we will show that $\alpha_{\varphi}(x)$ is non-increasing from $0$ to $1$, which implies that for $x \in [0,1)$, we have $\alpha_{\varphi}(x) \geq  \min\{\alpha_{\varphi}(\lfloor x \rfloor), \alpha_{\varphi}(\lceil x \rceil)\}$.
Recall that $\varphi$, by definition, is linear between integers. Hence, the fact that $\varphi(0) = 0$ and $\varphi(1) = 1$, gives us $\varphi(x) = x$ for all $x \in [0,1]$. Therefore, 
  \[\alpha_{\varphi}(x) = \frac{\mathbb{E}[\varphi(\Poi(x))]}{x} = e^{-x}\sum_{k=1}^{+\infty}\frac{\varphi(k)}{k}\frac{x^{k-1}}{(k-1)!} = e^{-x}\sum_{k=0}^{+\infty}\frac{\varphi(k+1)}{k+1}\frac{x^k}{k!} \ . \]
  In particular, $\alpha_{\varphi}(x)$ is well-defined at $0$ and $\alpha_{\varphi}(0) = e^{-0}\sum_{k=0}^{+\infty}\frac{\varphi(k+1)}{k+1}\frac{0^k}{k!} = 1$. Now, consider the derivative
   \[\alpha_{\varphi}'(x) = e^{-x}\left(-\sum_{k=0}^{+\infty}\frac{\varphi(k+1)}{k+1}\frac{x^k}{k!} + \sum_{k=1}^{+\infty}\frac{\varphi(k+1)}{k+1}\frac{x^{k-1}}{(k-1)!} \right) = e^{-x}\sum_{k=0}^{+\infty}\left(\frac{\varphi(k+2)}{k+2} - \frac{\varphi(k+1)}{k+1} \right)\frac{x^k}{k!} \ .\]

Note that $\frac{\varphi(k+2)}{k+2} - \frac{\varphi(k+1)}{k+1} = \frac{\varphi(k+2) - \varphi(0)}{(k+2)-0} - \frac{\varphi(k+1) - \varphi(0)}{(k+1)-0}  \leq 0$; the last inequality follows from the concavity of $\varphi$. Hence, $\alpha_{\varphi}'(x) \leq 0$. That is, $\alpha_{\varphi}(x)$ is non-increasing from $0$ to $1$.
\end{proof}

\begin{prop}
For any $\varepsilon >0$,  the bound $1 - \alpha_{\varphi}(x) \leq \varepsilon$ holds for all $x \geq \left(\frac{6}{\varepsilon}\right)^4$.
\label{prop:ConvergenceRate}
\end{prop}
\begin{proof}
Write $X \sim \Poi(x)$ and note that $\mathbb{P}(X \leq x(1-\delta(x))) \leq \exp(-\frac{x\delta(x)^2}{2(1+\delta(x))})$, for any positive function $\delta(\cdot)$ which satisfies $\delta(x) < 1$, for all $x >1$; see, e.g., \cite{PoissonBound}. Therefore, 
    \begin{align}
      \mathbb{E}[\varphi(X)] &\geq  e^{-x}\sum_{k=\lceil x(1-\delta(x)) \rceil}^{+\infty}\varphi(k)\frac{x^k}{k!} \quad \text{since $\varphi$ nonnegative} \\
      &\geq   \varphi(x(1-\delta(x)))\sum_{k=\lceil x(1-\delta(x)) \rceil}^{+\infty}e^{-x}\frac{x^k}{k!}  \quad \text{since $\varphi$ nondecreasing} \\
      &\geq  \varphi(x(1-\delta(x)))(1-\mathbb{P}(X \leq x(1-\delta(x))))\\
      &\geq  \varphi(x(1-\delta(x)))\left(1-\exp(-\frac{x\delta(x)^2}{2(1+\delta(x))})\right)\ .  \label{ineq:exp-soup}
    \end{align}

Next, we will show that $\frac{\varphi(x(1-\delta(x)))}{\varphi(x)} \geq 1-\frac{\delta(x)+\frac{1}{x}}{1-\delta(x)}$. Towards this end, we will first bound $\varphi(x+y)-\varphi(x)$ in terms of $w^x_k = \varphi(x+k)-\varphi(x+k-1)$, which constitutes a non-increasing sequence (since $\varphi$ is concave):

    \[ \varphi(x+y) - \varphi(x) \leq \varphi(x+\lfloor y \rfloor + 1)-\varphi(x+\lfloor y \rfloor) + \sum_{k=1}^{\lfloor y \rfloor} w^x_k \leq (\lfloor y \rfloor + 1) w^x_1. \]

    Applying this bound to $x(1-\delta(x))$ and $x\delta(x)$ gives us
   
      \begin{align}
        1 - \frac{\varphi(x(1-\delta(x)))}{\varphi(x)} &=  \frac{\varphi(x)-\varphi(x(1-\delta(x)))}{\varphi(x)} \leq \frac{(\lfloor x\delta(x) \rfloor+1) w^{x(1-\delta(x))}_1}{\varphi(x)} \nonumber \\
        &\leq  \frac{x\delta(x)+1}{\varphi(x)} \ \frac{\varphi(x(1-\delta(x)))}{x(1-\delta(x))} \leq \frac{x\delta(x)+1}{x(1-\delta(x))} = \frac{\delta(x)+\frac{1}{x}}{1-\delta(x)}. \label{ineq:phi-soup}    \end{align}

Here, $w^{x(1-\delta(x))}_1 = \frac{\varphi(x(1-\delta(x))+1) - \varphi(x(1-\delta(x)))}{(x(1-\delta(x))+1)-(x(1-\delta(x)))} \leq \frac{\varphi(x(1-\delta(x)))-\varphi(0)}{x(1-\delta(x))-0}= \frac{\varphi(x(1-\delta(x)))}{x(1-\delta(x))}$ follows from the concavity of $\varphi$ and $\frac{\varphi(x(1-\delta(x)))}{\varphi(x)}\ \leq 1$ from the fact that $\varphi$ is nondecreasing. 

Inequalities (\ref{ineq:exp-soup}) and (\ref{ineq:phi-soup}) lead to following upper bound on $1 - \alpha_{\varphi}(x)$ in terms of $\delta(x)$: 
    \begin{align}
        1 - \alpha_{\varphi}(x) & = 1 - \frac{\mathbb{E}[\varphi(\Poi(x))]}{\varphi(x)} \leq 1 - \left(1-\frac{\delta(x) + \frac{1}{x}}{1-\delta(x)}\right)\left(1-\exp(-\frac{x\delta(x)^2}{2(1+\delta(x))})\right) \nonumber \\
         & \leq  \frac{\delta(x) + \frac{1}{x}}{1-\delta(x)} +  \exp(-\frac{x\delta(x)^2}{2(1+\delta(x))}). \label{ineq:delta-bound}
    \end{align}

Specifically setting $\delta(x)=x^{-\frac{1}{4}}$, we have (for all $x \geq 16$): $\delta(x) \leq \frac{1}{2}$, $\frac{1}{x} \leq x^{-\frac{1}{4}}$, and $\exp(-\frac{x\delta(x)^2}{2(1+\delta(x))}) \leq \exp(-\frac{\sqrt{x}}{4}) \leq 2x^{-\frac{1}{4}}$. Hence, inequality (\ref{ineq:delta-bound}) reduces to 
\begin{align*}
1 - \alpha_{\varphi}(x) \leq \frac{2x^{-\frac{1}{4}}}{1-\frac{1}{2}} + 2x^{-\frac{1}{4}} \leq 6x^{-\frac{1}{4}} \qquad \text{for all $x \geq 16$.}
\end{align*}

If $\varepsilon \geq 1$, we have $1 - \alpha_{\varphi}(x) \leq 1 \leq \varepsilon$. Otherwise, we have that $\left(\frac{6}{\varepsilon}\right)^4 \geq 6^4 \geq 16$. Therefore, given any $\varepsilon >0$, for all $x \geq \left(\frac{6}{\varepsilon}\right)^4$ we have $1 - \alpha_{\varphi}(x) \leq \varepsilon$.
\end{proof}

\begin{prop}
    We have that $\alpha_{\varphi} = \inf_{x \in \mathbb{R}_+} \alpha_{\varphi}(x) = \min_{x \in \mathbb{N}^*} \alpha_{\varphi}(x)$.
  \label{prop:minNotInf}
\end{prop}

\begin{proof}
  Thanks to Proposition \ref{prop:minInt}, we have that $\inf_{x \in \mathbb{R}^+} \alpha_{\varphi}(x) = \inf_{x \in \mathbb{N}^*} \alpha_{\varphi}(x)$, and thanks to Proposition \ref{prop:ConvergenceRate}, since $\alpha_{\varphi}(x) \leq 1$, we have that $\inf_{x \in \mathbb{N}^*} \alpha_{\varphi}(x) = \min_{x \in \mathbb{N}^*} \alpha_{\varphi}(x)$.
\end{proof}

\begin{prop}
  $C^{\varphi}$ is submodular, its curvature is at most $c = 1 - (\varphi(m) - \varphi(m-1))$ and it cannot be improved for a general instance with $m$ cover sets.
  \label{prop:SubCurv}
\end{prop}

\begin{proof}

  We use the following lemma which is trivial to prove:

  \begin{lem}[Properties of $\abs{S}_a = \abs{\set{i \in S : a \in T_i}}$.]
    We have:
  \begin{enumerate}
  \item $\abs{S}_a \leq \abs{S}$,
  \item $\abs{S \cup S'}_a\leq \abs{S}_a + \abs{S'}_a$. In particular, if $S \subseteq T$ then $\abs{S}_a \leq \abs{T}_a$ and $\abs{S\cup\set{x}}_a \leq \abs{S}_a + 1$,
  \item If $S \subseteq T$, $x \not\in T$ then $\abs{S}_a = \abs{T}_a \Rightarrow \abs{S\cup\set{x}}_a = \abs{T\cup\set{x}}_a$.
  \end{enumerate}
  \label{lem:ke}
\end{lem}
  
  Let us show first the submodularity of $C^{\varphi}$. Let $S \subseteq T \subseteq [m]$ and $x \not\in T$:
  \begin{equation}
    \begin{aligned}
      &&& C^{\varphi}(S \cup \set{x}) - C^{\varphi}(S) - (C^{\varphi}(T \cup \set{x}) - C^{\varphi}(T))= \\
      &=&& \sum_{a \in [n]} w_a[\varphi(\abs{S\cup\set{x}}_a) - \varphi(\abs{S}_a) - (\varphi(\abs{T\cup\set{x}}_a) - \varphi(\abs{T}_a))]\ .\\
    \end{aligned}
  \end{equation}

  Let us call $g(a) := \varphi(\abs{S\cup\set{x}}_a) - \varphi(\abs{S}_a) - (\varphi(\abs{T\cup\set{x}}_a) - \varphi(\abs{T}_a))$:
  \begin{enumerate}
  \item If $\abs{T}_a = \abs{S}_a$ then thanks to Lemma \ref{lem:ke}, we have that $\abs{T\cup\set{x}}_a = \abs{S\cup\set{x}}_a$, so $g(a) = 0$
  \item Else, we have that $\abs{T}_a > \abs{S}_a$:
    \begin{enumerate}
    \item If $\abs{S\cup\set{x}}_a = \abs{S}_a$, then we add elements of $T-S$ using Lemma \ref{lem:ke} to get that $\abs{T\cup\set{x}}_a = \abs{T}_a$, so $g(a)=0$ in that case.
    \item Else $\abs{S\cup\set{x}}_a \not= \abs{S}_a$. So with $\abs{S}_a = k$, we get that $\abs{S\cup\set{x}}_a = k+1$ and $\abs{T}_a > \abs{S}_a$ so $\abs{T}_a \geq k+1$.

      \begin{enumerate}
      \item If $\abs{T\cup\set{x}}_a = \abs{T}_a$, then $g(a) =  \varphi(k+1) - \varphi(k) \geq 0$ since $\varphi$ is nondecreasing.
      \item Else $\abs{T\cup\set{x}}_a \not= \abs{T}_a$ so with $\abs{T}_a = \ell$ with $\ell \geq k+1$, we get that $\abs{T}_a = \ell+1$. So we have that:
        \begin{equation}
          \begin{aligned}
            g(a) &=&& \varphi(k+1) - \varphi(k) - (\varphi(\ell+1) - \varphi(\ell))\\
            &=&& \frac{\varphi(k+1) - \varphi(k)}{(k+1) - k} - \frac{\varphi(\ell+1) - \varphi(\ell)}{(\ell+1) - \ell} \geq 0\ ,
          \end{aligned}
        \end{equation}
        by concavity of $\varphi$: its slopes are nonincreasing.
      \end{enumerate}
    \end{enumerate}
  \end{enumerate}
  So in all cases, we have $g(a) \geq 0$ so $ C^{\varphi}(S \cup \set{x}) - C^{\varphi}(S) - (C^{\varphi}(T \cup \set{x}) - C^{\varphi}(T)) \geq 0$: $C^{\varphi}$ is submodular.

  Let us now compute its curvature:
  \[c = 1 - \min_{i \in [m]} \frac{C^{\varphi}([m]) - C^{\varphi}([m]-\set{i})}{C^{\varphi}(\set{i}) - C^{\varphi}(\emptyset)}\ .\]

  Let $i \in [m]$ fixed:
  \begin{equation}
    \begin{aligned}
      \frac{C^{\varphi}([m]) - C^{\varphi}([m]-\set{i})}{C^{\varphi}(\set{i}) - C^{\varphi}(\emptyset)} &=&& \frac{\sum_{a \in [n]} w_a[\varphi(\abs{[m]}_a) - \varphi(\abs{[m]-\set{i}}_a)]}{\sum_{a \in [n]}w_a[\varphi(\abs{\set{i}}_a) - \varphi(\abs{\emptyset}_a)]}\\
      &=&& \frac{\sum_{a \in T_i} w_a[\varphi(\abs{[m]}_a) - \varphi(\abs{[m]-\set{i}}_a)]}{\sum_{a \in T_i} w_a}\\
      &=&& \frac{\sum_{a \in T_i} w_a[\varphi(\abs{[m]}_a) - \varphi(\abs{[m]}_a-1)]}{\sum_{a \in T_i} w_a} \text{ since } a \in T_i \ .\\
    \end{aligned}
  \end{equation}
  But $\abs{[m]}_a \leq m$ and $\varphi$ concave, so $\varphi(\abs{[m]}_a)) - \varphi(\abs{[m]}_a-1) \geq \varphi(m) - \varphi(m-1)$ for all $a \in [n]$. As a consequence we have that:

  \[\frac{C^{\varphi}([m]) - C^{\varphi}([m]-\set{i})}{C^{\varphi}(\set{i}) - C^{\varphi}(\emptyset)} \geq \varphi(m) - \varphi(m-1)\ .\]
  and this lower bound is true for its minimum over $i$. Thus we get that $c \leq 1 - (\varphi(m) - \varphi(m-1))$.
  Also one can find instances for all $m$ such that this bound is tight: take $T_1 =\set{a}$ and $\forall j \in [m], a \in T_j$ for instance.
\end{proof}

\begin{prop}
  Let $\ell \in \mathbb{N}^*$. if $\forall x \geq \ell, \varphi(x) = \varphi(\ell) +a(x-\ell)$ for some $0 \leq a \leq \varphi(\ell)-\varphi(\ell-1)$, then $\alpha_{\varphi}(x)$ is nondecreasing from $\ell$ to $+\infty$ and:
  \[\alpha_{\varphi}(x) =  \frac{\varphi(\ell) + a(x-\ell)}{\varphi(x)} - \frac{e^{-x}}{\varphi(x)}\left(\sum_{k=0}^{\ell} \left(\varphi(\ell) + a(x-\ell) - \varphi(k)\right)\frac{x^k}{k!} - a\frac{x^{\ell+1}}{\ell!} \right) \ .\]
  In particular, $\alpha_{\varphi} = \min_{x \in [\ell]} \alpha_{\varphi}(x)$, and the argmin can be computed numerically.
  \label{prop:phiLinear}
\end{prop}

\begin{proof}
  One can compute a closed form value for $\alpha_{\varphi}(x)$ using the fact that $\varphi$ is linear from $\ell$:
  \begin{align*}
    \alpha_{\varphi}(x) &= \frac{e^{-x}}{\varphi(x)}\sum_{k=0}^{+\infty}\varphi(k)\frac{x^k}{k!} = \frac{e^{-x}}{\varphi(x)}\left(\sum_{k=0}^{\ell}\varphi(k)\frac{x^k}{k!} + \sum_{k=\ell+1}^{+\infty}\left(\varphi(\ell)+a(k-\ell)\right)\frac{x^k}{k!} \right)\\
    &= \frac{\varphi(\ell)-a\ell}{\varphi(x)} + \frac{e^{-x}}{\varphi(x)}\left(\sum_{k=0}^{\ell}\left(\varphi(k) - \varphi(\ell) + a\ell\right)\frac{x^k}{k!} + ax\sum_{k=\ell+1}^{+\infty}\frac{x^{k-1}}{(k-1)!} \right)\\
    &= \frac{\varphi(\ell)+a(x-\ell)}{\varphi(x)} + \frac{e^{-x}}{\varphi(x)}\left(\sum_{k=0}^{\ell}\left(\varphi(k) - \varphi(\ell) + a(\ell-x)\right)\frac{x^k}{k!} + ax\frac{x^{\ell}}{\ell!} \right) \ ,\\
  \end{align*}
  and thus we get:
  
  \[ \alpha_{\varphi}(x) = \frac{\varphi(\ell) + a(x-\ell)}{\varphi(x)} - \frac{e^{-x}}{\varphi(x)}\left(\sum_{k=0}^{\ell} \left(\varphi(\ell) + a(x-\ell) - \varphi(k)\right)\frac{x^k}{k!} - a\frac{x^{\ell+1}}{\ell!} \right) \ .\]

  Let us show that it is nondecreasing from $\ell$ to $+\infty$ by computing its derivative. Indeed, for $x \geq \ell$, we have that $\varphi(x) = \varphi(\ell) + a(x-\ell)$ and $\varphi'(x) = a$, so:

  \[
    \alpha_{\varphi}(x) = 1 - e^{-x}\left(\sum_{k=0}^{\ell} \frac{x^k}{k!} - \frac{1}{\varphi(x)}\left(\sum_{k=0}^{\ell} \varphi(k)\frac{x^k}{k!} + a\frac{x^{\ell+1}}{\ell!}\right) \right) \ .
  \]

  Thus:

  \begin{align*}
    \alpha_{\varphi}'(x) &= e^{-x}\left(\sum_{k=0}^{\ell} \frac{x^k}{k!} - \frac{1}{\varphi(x)}\left(\sum_{k=0}^{\ell} \varphi(k)\frac{x^k}{k!} + a\frac{x^{\ell+1}}{\ell!}\right) \right)\\
    &- e^{-x}\left(\sum_{k=0}^{\ell-1} \frac{x^k}{k!}  - \frac{1}{\varphi(x)}\left(\sum_{k=0}^{\ell-1} \varphi(k+1)\frac{x^k}{k!} + a(\ell+1)\frac{x^{\ell}}{\ell!}\right) + \frac{a}{\varphi(x)^2}\left(\sum_{k=0}^{\ell} \varphi(k)\frac{x^k}{k!} + a\frac{x^{\ell+1}}{\ell!}\right)\right)\\
    &= e^{-x}\left(\frac{x^{\ell}}{\ell!} - \frac{1}{\varphi(x)}\left(\sum_{k=0}^{\ell-1}(\varphi(k)- \varphi(k+1))\frac{x^k}{k!} + \varphi(\ell)\frac{x^{\ell}}{\ell!} +  (a(x-\ell)-a)\frac{x^{\ell}}{\ell!}\right)\right)\\
    &- e^{-x}\left(\frac{a}{\varphi(x)^2}\left(\sum_{k=0}^{\ell} \varphi(k)\frac{x^k}{k!} + a\frac{x^{\ell+1}}{\ell!}\right)\right)\\
    &= e^{-x}\left(\frac{x^{\ell}}{\ell!} + \frac{1}{\varphi(x)}\left(\sum_{k=0}^{\ell-1}(\varphi(k+1)- \varphi(k))\frac{x^k}{k!}\right) - \frac{\varphi(x)-a}{\varphi(x)}\frac{x^{\ell}}{\ell!}  - \frac{a}{\varphi(x)^2}\left(\sum_{k=0}^{\ell} \varphi(k)\frac{x^k}{k!} + a\frac{x^{\ell+1}}{\ell!}\right)\right)\\
    &= e^{-x}\left(\frac{1}{\varphi(x)}\left(\sum_{k=0}^{\ell-1}(\varphi(k+1)- \varphi(k))\frac{x^k}{k!}\right) + \frac{a}{\varphi(x)^2}\left(\varphi(x)\frac{x^{\ell}}{\ell!} - \sum_{k=0}^{\ell} \varphi(k)\frac{x^k}{k!} - a\frac{x^{\ell+1}}{\ell!}\right)\right) \ .\\
  \end{align*}
  
  If $a=0$, then it is nonnegative since $\varphi$ nondecreasing and nonnegative. Otherwise, suppose that $a > 0$. Then:

\begin{align*}
    \alpha_{\varphi}(x) &= \frac{ae^{-x}}{\varphi(x)^2} \left(\sum_{k=0}^{\ell-1}\left(\varphi(x)\frac{\varphi(k+1)- \varphi(k)}{a} - \varphi(k)\right)\frac{x^k}{k!} +\left(\varphi(x)-\varphi(\ell) - ax\right)\frac{x^{\ell}}{\ell!}\right)\\
    &\geq \frac{ae^{-x}}{\varphi(x)^2} \left(\sum_{k=0}^{\ell-1}\left(\varphi(x)- \varphi(k)\right)\frac{x^k}{k!} -a\frac{x^{\ell}}{(\ell-1)!}\right) \ ,
\end{align*}
since $\frac{\varphi(k+1)- \varphi(k)}{a} \geq \frac{\varphi(k+1)- \varphi(k)}{\varphi(\ell)-\varphi(\ell-1)} \geq 1$ by concavity of $\varphi$. Thus:
\begin{align*}
  \alpha_{\varphi}(x) &\geq \frac{ae^{-x}}{\varphi(x)^2} \left(\sum_{k=0}^{\ell-1}\left(\varphi(\ell)- \varphi(k)\right)\frac{x^k}{k!} +a\left((x-\ell)\sum_{k=0}^{\ell-1}\frac{x^k}{k!}-\frac{x^{\ell}}{(\ell-1)!}\right)\right) \ ,
  \end{align*}
but:
\[(x-\ell)\sum_{k=0}^{\ell-1}\frac{x^k}{k!}-\frac{x^{\ell}}{(\ell-1)!} = \ell\sum_{k=1}^{\ell}\frac{x^k}{k!} -\ell\sum_{k=0}^{\ell-1}\frac{x^k}{k!}-\frac{x^{\ell}}{(\ell-1)!} = - \ell \ ,\]
so:
 \begin{align*}
   \alpha_{\varphi}'(x) &\geq \frac{ae^{-x}}{\varphi(x)^2} \left(\sum_{k=0}^{\ell-1}\left(\varphi(\ell)- \varphi(k)\right)\frac{x^k}{k!} -a\ell \right)\\
   &\geq \frac{ae^{-x}}{\varphi(x)^2} \left(\sum_{k=0}^{\ell-1}\left(\varphi(\ell)- \varphi(k)\right)\frac{x^k}{k!} -(\varphi(\ell)-\varphi(\ell-1))\ell \right)\\
    &\geq \frac{ae^{-x}}{\varphi(x)^2} \left(\sum_{k=0}^{\ell-1}\left(\varphi(\ell)- \varphi(k)\right)\frac{x^k}{k!} - \sum_{k=0}^{\ell-1}(\varphi(\ell)-\varphi(\ell-1))\frac{x^k}{k!}\right) \quad \text{since $\frac{x^k}{k!} \geq \frac{\ell^k}{k!} \geq 1$ for $k \leq \ell$}\\
   &= \frac{ae^{-x}}{\varphi(x)^2} \left(\sum_{k=0}^{\ell-1}\left(\varphi(\ell-1)- \varphi(k)\right)\frac{x^k}{k!}\right) \geq 0  \quad \text{since $\varphi$ nondecreasing} \ .
 \end{align*}

 Thus, $\alpha_{\varphi}(x)$ is nondecreasing from $\ell$ to $+\infty$, and we get that $\alpha_{\varphi} = \min_{x \in [\ell]} \alpha_{\varphi}(x)$.
\end{proof}

\begin{prop}
  The Poisson concavity ratio $\alpha_{\varphi}$ is always greater than or equal to the curvature-dependent ratio defined in $\cite{SVW17}$: if $\varphi$ is linear from $m$ with slope $1-c = \varphi(m) - \varphi(m-1)$, then we have $\alpha_{\varphi} \geq 1 - ce^{-1}$.
  \label{prop:BetterRatio}
\end{prop}

\begin{proof}
  Note that by Proposition \ref{prop:SubCurv}, the curvature of $C^{\varphi}$ is equal to $c$, so the efficiency of the algorithm described in $\cite{SVW17}$ is indeed $ 1 - ce^{-1}$.
  Thanks to Proposition \ref{prop:phiLinear}, we have that $\alpha_{\varphi} = \min_{x \in [\ell]} \alpha_{\varphi}(x)$, so we only have to show that:
  \[\min_{\ell \in [m]} \alpha_{\varphi}(\ell) \geq 1 - ce^{-1} \ .\]

  Let us denote by $\varphi^{(\ell,a)}$ the function which is equal to $\varphi$ for $k \leq \ell$ and linear from $\ell$ with nonnegative coefficient $a$: $\forall k \geq \ell, \varphi^{(\ell,a)}(k) = \varphi(\ell) + a(k-\ell)$. Note that we ask that $0 \leq a \leq \varphi(\ell)-\varphi(\ell-1)$ in order to $\varphi^{(\ell,a)}$ to be nondecreasing concave, and $\ell \geq 1$.

This is done in two steps:

\begin{enumerate}
\item Let $1 \leq \ell \leq m$, then:
  \[\alpha_{\varphi}(\ell) = \frac{\mathbb{E}[\varphi(\Poi(\ell))]}{\varphi(\ell)} = \frac{\mathbb{E}[\varphi(\Poi(\ell))]}{\varphi^{(\ell,1-c)}(\ell)} \geq \frac{\mathbb{E}[\varphi^{(\ell,1-c)}(\Poi(\ell))]}{\varphi^{(\ell,1-c)}(\ell)} = \alpha_{\varphi^{(\ell,1-c)}}(\ell) \ , \]
  since $\varphi^{(\ell,1-c)}(x) \leq \varphi(x)$ for all $x$. Note that we have $\varphi(\ell)-\varphi(\ell-1) \geq \varphi(m)-\varphi(m-1) = 1-c$ by concavity of $\varphi$. So, we only have to show that for all $1 \leq \ell \leq m$, we have $\alpha^{(\ell,1-c)} := \alpha_{\varphi^{(\ell,1-c)}}(\ell) \geq 1 - ce^{-1}$.

\item Let us show that $\alpha^{(\ell,1-c)} := \alpha_{\varphi^{(\ell,1-c)}}(\ell) \geq 1 - ce^{-1}$ for $1 \leq \ell \leq m$.

  Using the closed-form expression of Proposition \ref{prop:phiLinear} on $\varphi^{(\ell,1-c)}$ evaluated at $\ell$, one gets:
  \[ \alpha^{(\ell,1-c)} = \alpha_{\varphi^{(\ell,1-c)}}(\ell) = 1 - e^{-\ell}\left(\sum_{k=0}^{\ell-1}\left(\frac{\varphi(\ell)-\varphi(k)}{\varphi(\ell)}\right)\frac{\ell^k}{k!}  - \frac{1-c}{\varphi(\ell)}\frac{\ell^{\ell+1}}{\ell!} \right) \ .\]
  
The worst case occurs when $\varphi^{(\ell,1-c)}$ is linear between $1$ and $\ell$, which we call $\varphi_{\text{lin}}^{(\ell,1-c)}$. Indeed, if we call $b:=\frac{\varphi(\ell)-1}{\ell-1}$, then for $1 \leq k \leq \ell$, we have that $\varphi_{\text{lin}}^{(\ell,1-c)}(k) = 1 + b(k-1)$. But:

\[\sum_{k=0}^{\ell-1}\left(\frac{\varphi(\ell)-\varphi(k)}{\varphi(\ell)}\right)\frac{\ell^k}{k!} \leq 1 + \sum_{k=1}^{\ell-1}\left(\frac{\varphi(\ell)- \left(1+b(k-1)\right)}{\varphi(\ell)}\right)\frac{\ell^k}{k!} \ ,\]
since $\varphi(k) \geq 1 + b(k-1)$, because $\frac{\varphi(k)-\varphi(1)}{k-1} \geq \frac{\varphi(\ell)-\varphi(1)}{\ell-1} = b$ by concavity of $\varphi$. In that case, the expression can be simplified:

\begin{align*}
  \alpha^{(\ell,1-c)} \geq \alpha_{\varphi_{\text{lin}}^{(\ell,1-c)}}(\ell) &=  1-e^{-\ell}\left(1+ \sum_{k=1}^{\ell-1}\left(\frac{b(\ell-k)}{\varphi(\ell)}\right)\frac{\ell^k}{k!} -\frac{1-c}{\varphi(\ell)}\frac{\ell^{\ell+1}}{\ell!}\right)\\
  &= 1-e^{-\ell}\left(1 + \frac{b\ell}{\varphi(\ell)}\sum_{k=1}^{\ell-1}\frac{\ell^k}{k!} - \frac{b\ell}{\varphi(\ell)}\sum_{k=1}^{\ell-1}\frac{\ell^{k-1}}{(k-1)!} -\frac{1-c}{\varphi(\ell)}\frac{\ell^{\ell+1}}{\ell!}\right)\\
  &= 1-\frac{e^{-\ell}}{\varphi(\ell)}\left(\varphi(\ell) + b\ell\left(\frac{\ell^{\ell-1}}{(\ell-1)!} - 1\right) -(1-c)\frac{\ell^{\ell+1}}{\ell!}\right)\\
     &= 1-\frac{e^{-\ell}}{\varphi(\ell)}\left(1 + b(\ell-1) + b\left(\frac{\ell^{\ell}}{(\ell-1)!} - \ell\right) -(1-c)\frac{\ell^{\ell+1}}{\ell!}\right)\\
  &= 1 - e^{-\ell}\frac{1-b + (b-(1-c))\frac{\ell^{\ell+1}}{\ell!}}{\varphi(\ell)} \ .
\end{align*}

We have also that $b \geq \varphi(\ell)-\varphi(\ell-1) \geq 1-c$ since $\varphi$ concave. As a function of $(b-(1-c))$ for $c$ fixed, we get $g(x) := 1 - e^{-\ell}\frac{c + x\left(\frac{\ell^{\ell+1}}{\ell!}-1\right)}{1 + (x+(1-c))(\ell-1)}$. In particular, we have that $\alpha_{\varphi_{\text{lin}}^{(\ell,(1-c))}}(\ell) = g(b-(1-c))$, since $\varphi(\ell) = 1 + b(\ell-1)$. We have that $g'(x) = -e^{-\ell}\frac{\ell\left(\frac{\ell^{\ell}}{\ell!}-1\right)+(1-c)\frac{\ell^{\ell+1}}{\ell!}(\ell-1)}{(1 + (x+(1-c))(\ell-1))^2} \leq 0$, so $g$ is nonincreasing: it is thus enough to show that $g(c) \geq 1-ce^{-1}$ to get the result, since $\alpha^{(\ell,1-c)} \geq g(b-(1-c)) \geq g(c) \geq  1-ce^{-1}$. But:

\[ g(c) = 1 - \frac{c\frac{\ell^{\ell+1}}{\ell!}}{1 + \ell-1}e^{-\ell} = 1 - c\frac{\ell^{\ell}}{\ell!}e^{-\ell} \geq 1-ce^{-1} \ , \]

since $\frac{\ell^\ell}{\ell!}e^{-\ell}$ is a decreasing sequence.
\end{enumerate}
\end{proof}

\begin{prop}
  Let $F(x) := \mathbb{E}_{X \sim x}[C^{\varphi}(X)]$ for $x \in \set{0,1}^m$. We have an explicit formula for $F$:
  \[  F(x) = \sum_{a = 1}^n \sum_{k=0}^{m} \Big[\frac{1}{m+1}\sum_{\ell = 0}^{m} \omega_{m+1}^{-\ell k}\prod_{j \in [m] : a \in  T_j}(1 +(\omega_{m+1}^{\ell}-1)x_j)\Big]\varphi(k) \text{ with } \omega_{m+1} := \exp(\frac{2i\pi}{m+1})\ . \]
  Thus, $F$ is computable in polynomial time in $n$ and $m$.
  \label{prop:Fpoly}
\end{prop}

\begin{proof}
  Recall that $C^{\varphi}(S) = \sum_{a=1}^n C_a^{\varphi}(S)$, so by linearity of expectation we can focus on $\mathbb{E}_{X \sim x}[C_a^{\varphi}(X)]$. But $C_a^{\varphi}(X) = \varphi(\abs{X}_a)$ where $\abs{X}_a = \abs{\set{ i \in [m] : X_i = 1 \text{ and } a \in T_i}} \in [0,m]$. Thus:
  \[ \mathbb{E}_{X \sim x}[C_a^{\varphi}(X)] = \sum_{k=0}^{m} \mathbb{P}_{X \sim x}(\abs{X}_a=k)\varphi(k)\ .\]

  It remains to compute the distribution of $\abs{X}_a$. But $\abs{X}_a = \sum_{i \in [m] : a \in T_i} X_i$ and $X_i \sim \Ber(x_i)$. Thus, $\abs{X}_a \sim \Poi\Bin((x_i)_{i \in [m] : a \in T_i})$, which is known as the Poisson binomial law. Thanks to \cite{FW10}, we have that:
  \[ \mathbb{P}_{X \sim x}(\abs{X}_a=k) = \frac{1}{m+1}\sum_{\ell = 0}^{m} \omega_{m+1}^{-\ell k}\prod_{j \in [m] : a \in  T_j}(1 +(\omega_{m+1}^{\ell}-1)x_j)\ ,\]
  where $\omega_{m+1} := \exp(\frac{2i\pi}{m+1})$, and the result is proved.
\end{proof}

\begin{prop}
  We have that
  \[ \abs{\mathbb{E}[\varphi(\Bin(n,x/n))] - \mathbb{E}[\varphi(\Poi(x))]} \leq \frac{x \varphi(n)}{2n} + \frac{x^{n+1}}{n!} \ .\]
  In particular when $\varphi(n) = o(n)$:
  \[ \lim_{n \rightarrow  \infty} \mathbb{E}[\varphi(\Bin(n,x_{\varphi}/n))] = \mathbb{E}[\varphi(\Poi(x_{\varphi}))] = \alpha_{\varphi}\varphi(x_{\varphi}) \ . \]
  \label{prop:UnboundBinPoi}
\end{prop}

\begin{proof}
  Thanks to \cite{Barbour84, TF19}, we have that the total variation distance between $\Bin(n,x/n)$ and $\Poi(x)$ is bounded in the following way:
  \[ \Delta(\Bin(n,x/n),\Poi(x)) \leq \frac{1 - e^{-x}}{2x} n \cdot \Big(\frac{x}{n}\Big)^2 \leq \frac{x}{2n}\ .\]
  Thus with $B \sim \Bin(n,x/n)$ and $P \sim \Poi(x)$:
  \begin{equation}
    \begin{aligned}
      \abs{\mathbb{E}[\varphi(B)] - \mathbb{E}[\varphi(P)]} &=&&  \abs{\sum_{k=0}^{+\infty}\varphi(k)\mathbb{P}(B=k) - \sum_{k=0}^{+\infty}\varphi(k)\mathbb{P}(P=k)}\\
      &=&&  \abs{\sum_{k=0}^{+\infty}\varphi(k)(\mathbb{P}(B=k) - \mathbb{P}(P=k))}\\
      &\leq&& \sum_{k=0}^{+\infty}\varphi(k)\abs{\mathbb{P}(B=k) - \mathbb{P}(P=k)}\\
      &\leq&& \varphi(n)\Delta(\Bin(n,x/n),\Poi(x)) + \sum_{k=n+1}^{+\infty}\varphi(k)\mathbb{P}(P=k)\\
      &\leq&&\frac{x \varphi(n)}{2n}  + e^{-x}\sum_{k=n+1}^{+\infty}k\frac{x^k}{k!} \quad \text{since } \varphi(k) \leq k\\
      &=&& \frac{x \varphi(n)}{2n}  + xe^{-x}\sum_{k=n}^{+\infty}\frac{x^k}{k!}\\
      &\leq&& \frac{x \varphi(n)}{2n}  + \frac{x^{n+1}}{n!} \underset{n \rightarrow \infty}{\rightarrow} 0 \text{ when } \varphi(n) = o(n)\ ,
    \end{aligned}
  \end{equation}
  by a standard upper bound on the remainder of the exponential series.

\end{proof}

\begin{prop}
  The function $g : x \mapsto \mathbb{E}[\varphi(\Poi(x))]$ on $\mathbb{R}^+$ is $\mathcal{C}^{\infty}$ nondecreasing concave.
  \label{prop:PoiCon}
\end{prop}

\begin{proof}
  Since we have that $0 \leq \varphi(k) \leq k$ for $k \in \mathbb{N}$, in particular $g(x) = e^{-x}\sum_{k=0}^{+\infty} \varphi(k)\frac{x^k}{k!}$ is $\mathcal{C}^{\infty}$. It is thus enough to compute its first and second derivatives:

  \begin{equation}
    \begin{aligned}
      g'(x) &=&& -e^{-x}\sum_{k=0}^{+\infty}\varphi(k)\frac{x^k}{k!}+ e^{-x}\sum_{k=1}^{+\infty}\varphi(k)k\frac{x^{k-1}}{k!}\\
      &=&& -e^{-x}\sum_{k=0}^{+\infty}\varphi(k)\frac{x^k}{k!}+ e^{-x}\sum_{k=0}^{+\infty}\varphi(k+1)\frac{x^{k}}{k!}\\
      &=&& e^{-x}\sum_{k=0}^{+\infty}(\varphi(k+1) -\varphi(k))\frac{x^k}{k!} \ .
    \end{aligned}
  \end{equation}
  But $\varphi(k+1) -\varphi(k) \geq 0$ since $\varphi$ nondecreasing, so $g'(x) \geq 0$ and $g$ is nondecreasing.

  The calculus of $g''$ is the same where we replace $\varphi$ by $\psi(k) := \varphi(k+1) -\varphi(k)$ which is a nonincreasing function by concavity of $\varphi$. Thus:
  \[g''(x) = e^{-x}\sum_{k=0}^{+\infty}(\psi(k+1) -\psi(k))\frac{x^k}{k!} \leq 0\ .\]
  since $\psi(k+1) -\psi(k) \leq 0$, and so $g$ is concave.
\end{proof}

\begin{prop}
  The function $g_q : n \mapsto \mathbb{E}[\varphi(\Bin(n,q))]$ defined on $\mathbb{N}$ is nondecreasing concave. As a consequence, one can uses Jensen's inequality on the piecewise linear extension of $g_q$ which is also continuous.
  \label{prop:BinCon}
\end{prop}

\begin{proof}
  $\Bin(n,q) \leq_{\text{st}} \Bin(n+1,q)$ and we have that $\varphi$ is nondecreasing, so $\mathbb{E}[\varphi(\Bin(n,q))] \leq \mathbb{E}[\varphi(\Bin(n+1,q))]$, ie $g_q(n+1) - g_q(n) \geq 0$: $g_q$ is nondecreasing.

  We show then the concavity, ie. $g_q(n+2) - g_q(n+1) \leq g_q(n+1) - g_q(n)$. Call $\psi(x) = \varphi(x+1)-\varphi(x)$ which is nonincreasing since $\varphi$ concave. Let us take $X_{k,q} \sim \Bin(k,q)$. Then:

  \begin{equation}
    \begin{aligned}
      g_q(n+1) &=&& \mathbb{E}[\varphi(X_{n+1,q})]\\
      &=&& \sum_{i=0}^n \mathbb{E}[\varphi(X_{n,q}+X_{1,q})|X_{n,q}=i]\mathbb{P}(X_{n,q}=i)\\
      &=&& \sum_{i=0}^n \mathbb{E}[\varphi(i+X_{1,q}) - \varphi(i)]\mathbb{P}(X_{n,q}=i) + \sum_{i=0}^n \varphi(i)\mathbb{P}(X_{n,q}=i)\\
      &=&& \sum_{i=0}^n \mathbb{E}[\varphi(i+X_{1,q}) - \varphi(i)]\mathbb{P}(X_{n,q}=i) + g_q(n)\ .
    \end{aligned}
  \end{equation}

  Thus:

  \begin{equation}
    \begin{aligned}
      g_q(n+1) -g_q(n) &=&& \sum_{i=0}^n \mathbb{E}[\varphi(i+X_{1,q}) - \varphi(i)]\mathbb{P}(X_{n,q}=i)\\
      &=&& \sum_{i=0}^n q(\varphi(i+1) - \varphi(i))\mathbb{P}(X_{n,q}=i)\\
      &=&& q \mathbb{E}[\psi(\Bin(n,q))]\ .
    \end{aligned}
  \end{equation}

  Then thanks to the fact that  $\Bin(n,q) \leq_{\text{st}} \Bin(n+1,q)$ and $\psi$ is nonincreasing, we have that $\mathbb{E}[\psi(\Bin(n,q))] \geq \mathbb{E}[\psi(\Bin(n+1,q))]$, ie. $g_q(n+2) - g_q(n+1) \leq g_q(n+1) - g_q(n)$.
\end{proof}

\begin{prop}
  With $w_i := \varphi(i)-\varphi(i-1)$, we have:
  \[ \lim_{i \rightarrow +\infty} w_i = 0 \iff \varphi(n) = o(n)\ .\]
  \label{prop:thieleEqLim}
\end{prop}
\begin{proof}
  \begin{itemize}
  \item ($\Rightarrow$) Let $\epsilon > 0$, let us find a rank $N$ such that for $n \geq N$, $\frac{\varphi(n)}{n} \leq \epsilon$. Let $N_0$ the rank from which $w_i \leq \frac{\epsilon}{2}$ and $N_1$ the rank from which $\frac{1}{n} \sum_{i=1}^{N_0-1} w_i \leq \frac{\epsilon}{2}$. We have
    
    \begin{equation}
      \begin{aligned}
        \frac{\varphi(n)}{n} &=&& \frac{1}{n} \sum_{i=1}^{n} w_i \leq \frac{1}{n} \sum_{i=1}^{N_0-1} w_i + \frac{1}{n} \sum_{i=N_0}^{n-1} \frac{\epsilon}{2}\\
        &\leq&& \frac{\epsilon}{2} + \frac{\epsilon}{2} = \epsilon \text{ for } n \geq \max(N_0,N_1) =: N\ .
      \end{aligned}
    \end{equation}
  \item ($\Leftarrow$) Since $w_i = \varphi(i)-\varphi(i-1)$ is nonnegative and nonincreasing (respectively because $\varphi$ is nondecreasing and concave), then the sequence $w$ has a limit $L \geq 0$. But
    \[ \frac{\varphi(n)}{n} = \frac{1}{n} \sum_{i=1}^{n} w_i \geq L \ . \]
    
    Since the left hand side tends to $0$ by hypothesis, this means that $L=0$.
  \end{itemize}
\end{proof}

\begin{prop}
    If $w_i := \varphi(i) - \varphi(i-1)$ is \emph{geometrically dominant}, ie. $\forall i \in \mathbb{N}^*, \frac{w_i}{w_{i+1}} \geq \frac{w_{i+1}}{w_{i+2}}$, then $\alpha_{\varphi} = \alpha_{\varphi}(1)$.
  \label{prop:geoDominant}
\end{prop}

\begin{rk}
  Proposition \ref{prop:geoDominant} and in particular its proof uses similar ideas to the sketch provided in \cite{DMMS20}.
\end{rk}

  \begin{proof}
  Let $g(k) = \mathbb{E}[\varphi(\Poi(k))]$, and thus $\alpha_{\varphi}(k) = \frac{g(k)}{\varphi(k)}$. Let us show that for $k \in \mathbb{N}^*, \alpha_{\varphi}(k) \geq \alpha_{\varphi}(1)$, which will be enough to conclude. In order to show this, we will need the following lemmas:
  
  \begin{lem}
    $\forall k < i \in \mathbb{N}, w_i \geq w_{k+1}w_{i-k}$ and thus $\forall k,j \in \mathbb{N}, \varphi(k+j) - \varphi(k) \geq w_{k+1} \varphi(j)$.
    \label{lem:geoDominant}
  \end{lem}

  \begin{proof}
    We have that:
    \[w_i = \frac{w_i}{w_{i-1}}\frac{w_{i-1}}{w_{i-2}}\ldots\frac{w_{i-k+1}}{w_{i-k}}w_{i-k} \ .\]

    But for $j \in [k]$:

    \[ \frac{w_{i-j+1}}{w_{i-j}} \geq \frac{w_{(i-1)-j+1}}{w_{(i-1)-j}} \geq \ldots \geq \frac{w_{(k+1)-j+1}}{w_{(k+1)-j}} \ , \]

    since $w$ is geometrically dominant and $k+1 \leq i$. Thus applying this bound on each term of the previous product, we get:

    \[w_i \geq \frac{w_{k+1}}{w_{k}}\frac{w_{k}}{w_{k-1}}\ldots\frac{w_{2}}{w_{1}}w_{i-k} = \frac{w_{k+1}}{w_1} w_{i-k} = w_{k+1} w_{i-k} \ . \]

    In particular, $\forall k,j \in \mathbb{N}$, we get:
    \[\varphi(k+j) - \varphi(k) = \sum_{i=k+1}^{k+j} w_i \geq w_{k+1}\sum_{i=1}^{j} w_i = w_{k+1} \varphi(j) \ .\]
  \end{proof}
  
  \begin{lem}
    The piecewise linear extension on $[1,+\infty[$ of $w$, defined on integers by $w(k)=w_k$, is convex.
    \label{lem:wConvex}
  \end{lem}

  \begin{proof}
    We will show that $\forall k \in \mathbb{N}^*, w_{k+2}-w_{k+1} \geq w_{k+1}-w_k$ which implies the convexity of its piecewise linear extension on $[1,+\infty[$. For $k \in \mathbb{N}^*$ we have:

        \[ \frac{w_{k+1}}{w_{k+2}}-1 \leq \frac{w_{k+1}}{w_{k+2}}\Big(\frac{w_{k+1}}{w_{k+2}} - 1\Big) \leq \frac{w_{k+1}}{w_{k+2}}\Big(\frac{w_{k}}{w_{k+1}}-1\Big) = \frac{w_k-w_{k+1}}{w_{k+2}} \ , \]

        since $w$ is nonnegative nonincreasing (respectively $\varphi$ nondecreasing concave) and $\frac{w_{k+1}}{w_{k+2}} \leq \frac{w_{k}}{w_{k+1}}$ since $w$ is geometrically dominant. Then, multiplying by $-w_{k+2} \leq 0$ gives the expected result $w_{k+2}-w_{k+1} \geq w_{k+1}-w_k$.
  \end{proof}
  
  We have $g(k+1) = \mathbb{E}[\varphi(\Poi(k+1))] = \mathbb{E}[\varphi(\Poi(k)+\Poi(1))]$ since $\Poi(k+1) \sim \Poi(k)+\Poi(1)$. Thus:
  \begin{equation}
    \begin{aligned}
      g(k+1)-g(k) &=&& \mathbb{E}_{X,X' \sim \Poi(k), Y \sim \Poi(1)}[\varphi(X+Y) - \varphi(X')]\\
      &=&& \mathbb{E}_{X \sim \Poi(k), Y \sim \Poi(1)}[\varphi(X+Y) - \varphi(X)]\\
      &\geq&& \mathbb{E}_{X \sim \Poi(k), Y \sim \Poi(1)}[w_{X+1}\varphi(Y)] \quad \text{by Lemma \ref{lem:geoDominant}}\\
      &=&& \mathbb{E}_{X \sim \Poi(k)}[w(X+1)]\mathbb{E}_{Y \sim \Poi(1)}[\varphi(Y)] \quad \text{by independence of $w(X+1)$ and $\varphi(Y)$}.
    \end{aligned}
  \end{equation}
  
  Since $w$ is convex on $[1,+\infty[$ by Lemma \ref{lem:wConvex} and $\Poi(k)+1 \in [1,+\infty[$, we have that $\mathbb{E}[w(\Poi(k)+1)] \geq w(\mathbb{E}[\Poi(k)+1]) = w(k+1) = w_{k+1}$ thanks to Jensen's inequality. Note that $g(0) = \mathbb{E}[\varphi(\Poi(0))] = \varphi(0) = 0$. Then:

  \[g(k) = \sum_{i=0}^{k-1} g(i+1)-g(i) \geq  \Big(\sum_{i=0}^{k-1}w_{i+1}\Big)\mathbb{E}[\varphi(\Poi(1))] = \varphi(k)g(1) \ .\]

  Therefore:
  \[\alpha_{\varphi}(k) = \frac{g(k)}{\varphi(k)} \geq g(1) = \frac{g(1)}{\varphi(1)}=\alpha_{\varphi}(1) \ .\]
\end{proof}

\section{Calculations of $\alpha_{\varphi}$}
\begin{prop}
  For $\ell \in \mathbb{N}^*$ and $\varphi(j) = \min\{j,\ell\}$, we have that $\alpha_{\varphi}=1-\frac{\ell^{\ell}e^{-\ell}}{\ell!}$.
  \label{prop:lCover}
\end{prop}
\begin{proof}
  Thanks to Proposition \ref{prop:BetterRatio}, we have that $\alpha_{\varphi} = \min_{x \in \mathbb{N}^*} \alpha_{\varphi}(x)$. Let us compute $\mathbb{E}[\varphi(\Poi(x))]$:

  \begin{equation}
    \begin{aligned}
      \mathbb{E}[\varphi(\Poi(x))] &=&& e^{-x}\sum_{k=0}^{+\infty}\varphi(k)\frac{x^k}{k!}\\
      &=&& e^{-x}\sum_{k=0}^{\ell}k\frac{x^k}{k!} + e^{-x}\sum_{k=\ell+1}^{+\infty}\ell\frac{x^k}{k!}\\
      &=&& e^{-x}x \sum_{k=0}^{\ell-1}\frac{x^k}{k!} + \ell e^{-x}\sum_{k=\ell+1}^{+\infty}\frac{x^k}{k!}\\
      &=&& e^{-x}\Big[(x-\ell)\sum_{k=0}^{\ell-1}\frac{x^k}{k!} - \ell\frac{x^{\ell}}{\ell!}\Big] + \ell e^{-x}\sum_{k=0}^{+\infty}\frac{x^k}{k!}\\
      &=&& \ell - e^{-x}\Big[\frac{x^{\ell}}{(\ell-1)!} - (x-\ell)\sum_{k=0}^{\ell-1}\frac{x^k}{k!}\Big] \ .
    \end{aligned}
  \end{equation}

  Let us show that $\alpha_{\varphi}(x)$ takes its minimum in $\ell$, where we have indeed:
  \[ \alpha_{\varphi}(\ell) = \frac{1}{\ell}\Big( \ell - e^{-\ell}\Big[\frac{\ell^{\ell}}{(\ell-1)!} - (\ell-\ell)\sum_{k=0}^{\ell-1}\frac{\ell^k}{k!}\Big]\Big) = 1 - e^{-\ell}\frac{\ell^{\ell}}{\ell!} \ . \]

  Thanks to proposition \ref {prop:phiLinear}, $\alpha_{\varphi}(x)$ is nondecreasing from $\ell$ to $+\infty$. Suppose now that $\ell \geq 2$ (otherwise the result is already proved). Since $\alpha_{\varphi}(x)$ is differentiable, we have for $1 \leq x \leq \ell$:
  \begin{equation}
    \begin{aligned}
      \alpha_{\varphi}'(x) &=&& -\frac{\ell}{x^2} + e^{-x}\Big[\frac{x^{\ell-1}}{(\ell-1)!} -  \sum_{k=0}^{\ell-1}\frac{x^k}{k!}  + \ell\sum_{k=0}^{\ell-2}\frac{x^k}{(k+1)!} + \frac{\ell}{x}\Big]\\ 
      &-&& e^{-x}\Big[\frac{x^{\ell-2}}{(\ell-2)!} - \sum_{k=0}^{\ell-2}\frac{x^k}{k!} + \ell\sum_{k=0}^{\ell-3}\frac{x^k}{(k+2)k!} - \frac{\ell}{x^2}\Big]\\
      &=&&  \frac{\ell}{x}\Big(e^{-x}\Big(1+\frac{1}{x}\Big) - \frac{1}{x}\Big) + e^{-x}\Big[\Big(\frac{\ell}{\ell-1}-1\Big)\frac{x^{\ell-2}}{(\ell-2)!} + \ell\sum_{k=0}^{\ell-3}\Big(\frac{x^k}{(k+1)!} - \frac{x^k}{(k+2)k!} \Big) \Big]\\
      &=&&  \frac{\ell}{x}\Big(e^{-x}\Big(1+\frac{1}{x}\Big) - \frac{1}{x}\Big) + e^{-x}\Big[\frac{x^{\ell-2}}{(\ell-1)!} + \ell\sum_{k=0}^{\ell-3}\frac{x^k}{k!}\Big(\frac{1}{k+1} - \frac{1}{k+2} \Big) \Big]\\
      &=&&  \frac{\ell}{x}\Big(e^{-x}\Big(1+\frac{1}{x} + \frac{x^{\ell-1}}{\ell!} + x\sum_{k=0}^{\ell-3}\frac{x^k}{k!}\frac{1}{(k+1)(k+2)}\Big) - \frac{1}{x}\Big)\\
      &=&&  \frac{\ell e^{-x}}{x^2}\Big(\Big(1+x+ \frac{x^\ell}{\ell!} + \sum_{k=0}^{\ell-3}\frac{x^{k+2}}{(k+2)!}\Big) - e^x\Big)\\
      &=&& \frac{\ell e^{-x}}{x^2}\Big(\sum_{k=0}^{\ell} \frac{x^k}{k!} -e^x\Big) \leq 0\ .
    \end{aligned}
  \end{equation}
  since the partial sum of the exponential series is bounded by its total sum. Thus $\alpha_{\varphi}(x)$ is nonincreasing from $1$ to $\ell$, and nondecreasing after, so it takes indeed its minimum in $\ell$ and the proposition is proved.
\end{proof}

\begin{prop}
  For $p \in (0,1)$ and $\varphi(j)=\frac{1-(1-p)^j}{p}$, we have that $\alpha_{\varphi} = \frac{1 - e^{-p}}{p}$.
\label{prop:VTA}
\end{prop}

\begin{proof}
    By definition:
    
    \begin{equation}
      \begin{aligned}
        \alpha_{\varphi}(x) &=&& \frac{\mathbb{E}[\varphi(\Poi(x))]}{\varphi(x)} = \frac{\sum_{k=0}^{+\infty}\varphi(k)e^{-x}\frac{x^k}{k!}}{\varphi(x)}\\
        &=&& \frac{1-e^{-x}\sum_{k=0}^{+\infty}(1-p)^k\frac{x^k}{k!}}{p\varphi(x)}\\
        &=&& \frac{1 - e^{-x}e^{(1-p)x}}{p\varphi(x)} = \frac{1 - e^{-px}}{p\varphi(x)} \ .
      \end{aligned}
    \end{equation}

    If $x \geq 1$, $\alpha_{\varphi}(x) = \frac{1 - e^{-px}}{1-(1-p)^x} = \frac{1 - e^{-px}}{1-e^{-qx}}$ with $q = \ln(\frac{1}{1-p}) > 0$ and:
    \[\alpha_{\varphi}'(x) = \frac{pe^{-px}(1-e^{-qx}) - qe^{-qx}(1-e^{-px})}{(1-e^{-qx})^2} = \frac{pe^{-px} - qe^{-qx} + (q-p)e^{-(p+q)x}}{(1-e^{-qx})^2} \ .\]

    Let us take $t = \frac{p}{q} \in (0,1)$, since $q = \ln(\frac{1}{1-p}) > p > 0$, $x_1 = -px$ and $x_2 = -(p + q)x$. Then by strict convexity of the exponential function, we have:

    \[ e^{tx_1 + (1-t)x_2} < te^{x_1} + (1-t)e^{x_2} = \frac{pe^{-px} + (q-p)e^{-(p+q)x}}{q} \ .\]

    But $tx_1 + (1-t)x_2 = \frac{-p^2x}{q} + \frac{-(q-p)(p+q)x}{q}  = \frac{-p^2x}{q} + \frac{-(q^2x-p^2x)}{q} = -qx$, so we get $pe^{-px} - qe^{-qx} + (q-p)e^{-(p+q)x} > 0$, and $\alpha'_{\varphi}(x) > 0$. Thus, $\alpha_{\varphi}(x)$ increases from $1$ to infinity and takes its minimum in $1$:
    
    \[\alpha_{\varphi} = \alpha_{\varphi}(1) = \frac{1 - e^{-p}}{p} \ .\]
\end{proof}

\begin{prop}
  For $d \in (0,1)$ and $\varphi(j)=j^d$, we have that $\alpha_{\varphi} = e^{-1}\sum_{k=1}^{+\infty}\frac{k^d}{k!}$.
    \label{prop:dPower}
\end{prop}
\begin{proof}
  We have for $x \geq 1$:
    \[\alpha_{\varphi}(x) = \frac{\mathbb{E}[\Poi(x)^d]}{\varphi(x)} = \frac{e^{-x}\sum_{k=0}^{+\infty}k^d\frac{x^k}{k!}}{\varphi(x)} = e^{-x}\sum_{k=0}^{+\infty}k^d\frac{x^{k-d}}{k!} \ . \]

   Then:
   \begin{equation}
     \begin{aligned}
       \alpha_{\varphi}'(x) &=&& -\alpha_{\varphi}(x)+ e^{-x}\sum_{k=1}^{+\infty}(k-d)k^d\frac{x^{k-d-1}}{k!} \\
       &=&& -\alpha_{\varphi}(x)+ e^{-x}\sum_{k=0}^{+\infty}(k+1-d)(k+1)^d\frac{x^{k-d}}{(k+1)!}\\
       &=&& -\alpha_{\varphi}(x)+ e^{-x}\Big((1-d)x^{-d} +  \sum_{k=1}^{+\infty}(k+1-d)(k+1)^{d-1}\frac{x^{k-d}}{k!}\Big)\\
       &=&& e^{-x}x^{-d}\Big(1-d + \sum_{k=1}^{+\infty}(\frac{k+1-d}{k+1}(k+1)^d - k^d)\frac{x^k}{k!}\Big) \ . 
     \end{aligned}
   \end{equation}

   But the function $f(k) = \frac{k+1-d}{k+1}(k+1)^d - k^d$ is positive on $\mathbb{R}_+^*$, so we get that $\alpha_{\varphi}'(x) > 0$ for $x \geq 1$, thus $\alpha_{\varphi}(x)$ is increasing from $1$ to $+\infty$, so $\alpha_{\varphi} = \alpha_{\varphi}(1) = e^{-1}\sum_{k=1}^{+\infty}\frac{k^d}{k!}$.
\end{proof}

\section{NP-hardness of $\delta,h$-\textsc{AryGapLabelCover}}
\label{app:NPhardnessGap}
\begin{proof}[Proof of Proposition \ref{prop:AryGapLabelCover}]
  We reduce from the Label Cover problem described in \cite{DMMS20} which is known to be an NP-hard problem. The main idea of this reduction is the usual equivalence between bipartite graphs and hypergraphs.

  \begin{defi}
    A Label Cover instance $\mathcal{L} = (A,B,E,[L],[R],\set{\pi_e}_{e \in E})$ consists of a bi-regular bipartite graph $(A,B,E)$ with right degree $t$, alphabet sets $[L],[R]$ and for every edge $e \in E$, a constraint $\pi_e :[L] \rightarrow [R]$.
    A \emph{labeling} of $\mathcal{L}$ is a function $\sigma : A \rightarrow [L]$. We say that $\sigma$ \emph{strongly satisfies} a right vertex $v \in B$ if for every two neighbours $u,u'$ of $v$, we have $\pi_{(u,v)}(\sigma(u)) = \pi_{(u',v)}(\sigma(u'))$. Moreover, we say that  $\sigma$ \emph{weakly satisfies} a right vertex $v \in B$ if there exists two neighbours $u,u'$ of $v$ such that $\pi_{(u,v)}(\sigma(u)) = \pi_{(u',v)}(\sigma(u'))$.
  \end{defi}

 \begin{theo}[$\delta$-Gap-Label-Cover$(t,R)$ from \cite{DMMS20}]
  For any fixed integer $t \geq 2$ and fixed $\delta > 0$, there exists $R_0$ such that for any integer $R \geq R_0$, it is NP-hard for Label Cover instances $\mathcal{L} = (A,B,E,[L],[R],\set{\pi_e}_{e \in E})$ with right degree $t$ and right alphabet $[R]$ to distinguish between:
  
  \begin{itemize}
  \item[\textbf{YES:}] There exists a labeling $\sigma$ that \emph{strongly satisfies} all the right vertices.
  \item[\textbf{NO:}] No labeling \emph{weakly satisfies} more than $\delta$ fraction of the right vertices.
  \end{itemize}
\end{theo}

 The reduction is the following. From $\delta$-Gap-Label-Cover$(t,R)$, we take $h=t$ and the same parameters $\delta,R$. Given an instance $\mathcal{L} = (A,B,E,[L],[R],\set{\pi_e}_{e \in E})$, we take $\mathcal{G} = (A,E',[L],[R],\set{\pi'_{e',v}}_{e' \in E',v \in e'})$ with $E' = \set{N(b), b \in B}$ with $N(b)$ the set of neighbours of $b$ in $\mathcal{L}$, and $\pi'_{e',v} = \pi'_{N(b),v} := \pi_{v,b}$ since $v \in N(b)$. Since $(A,B,E)$ is bipartite and biregular, we get that our hypergraph has all hyperedges of size $h = \abs{N(b)} = t$, and that it is regular from the regular left degree of $(A,B,E)$. By construction, the notion of weakly and strongly satisfied is the same in both cases, as well as the labelings, and thus we have the NP-hardness of $\delta,h$-\textsc{AryGapLabelCover}.

 Note that both problems are in fact linearly equivalent since we could do the same reduction backwards.
 
\end{proof}

\section{Proof of existence of partitioning systems}
\label{app:Partitioning}

\begin{proof}[Proof of Proposition \ref{prop:Partitioning}]
    The existential proof is based on the probabilistic method. We take $\mathcal{P}_i$ an $h$-equi-sized uniform random $x_{\varphi}$-cover of $[n]$. Hence in the collection $\mathcal{P}_i=(P_{i,1},\ldots,P_{i,h})$, each of the $h$ subsets is of cardinality $\frac{x_{\varphi}n}{h}$. Write $\mathcal{P} = (\mathcal{P}_1,\ldots,\mathcal{P}_R)$. We have that for any $a \in [n], \mathbb{P}(a \in P_{i,j}) = \frac{x_{\varphi}}{h}$. Note that these events are independent for different $i$s.

    By construction, the first condition is fulfilled. Let us prove the second one.

    Fix $T \subseteq [R]$ and $\mathcal{Q} := \set{P_{i,j(i)} : i \in T}$ for some function $j : T \rightarrow [h]$. We have for $a \in [n]$:
    \[ \mathbb{E}[C_a^{\varphi}(\mathcal{Q})] =  \mathbb{E}[\varphi(\abs{\mathcal{Q}}_a)] = \mathbb{E}[\varphi(\abs{\set{i \in T: a \in P_{i,j(i)}}})]\ .\]

    But the random variables $\set{X^a_i := \mathbbm{1}_{a \in P_{i,j(i)}}}_{i \in T}$ are independent and $X^a_i \sim \Ber(\frac{x_{\varphi}}{h})$, so $X^a :=\abs{\set{i \in T: a \in P_{i,j(i)}}} = \sum_{i \in T} X^a_i \sim \Bin(\abs{T},\frac{x_{\varphi}}{h})$, and thus:
    \[\mathbb{E}[C_a^{\varphi}(\mathcal{Q})] = \mathbb{E}[\varphi(\Bin(\abs{T},\frac{x_{\varphi}}{h}))] =\psi^{\varphi}_{\abs{T},h}\ .\]

    Since $\abs{\mathcal{Q}}_a\leq \abs{\mathcal{Q}} \leq R$ and $\varphi$ nondecreasing, we have $0 \leq C_a^{\varphi}(\mathcal{Q}) \leq \varphi(R)$. We claim that we can apply a Chernoff-Hoeffding bound on $C^{\varphi}(\mathcal{Q}) = \sum_{a \in [n]} C_a^{\varphi}(\mathcal{Q})$ and get:

    \[ \mathbb{P}\Big( \abs{C^{\varphi}(\mathcal{Q}) -\psi^{\varphi}_{\abs{T},h} n} > \eta n\Big) \leq 2 \text{exp}\Big(-2\Big(\frac{\eta}{\varphi(R)}\Big)^2n\Big)\ .\]

    The random variables $\{ C_a^{\varphi}(\mathcal{Q}) \}_{a \in [n]}$ are not independent in general. However, they are negatively associated \cite{JP83}, and this is sufficient for the Chernoff-Hoeffding bound to hold as pointed out in \cite{DR98}, provided that $\eta \in (0,1)$. The set of random variables $\{ C_a^{\varphi}(\mathcal{Q})  \}_{a \in [n]}$ is said to be negatively associated if for any functions $f$ and $g$ either both increasing or both decreasing and any disjoint index sets $I,J \subseteq [n]$, we have:

    \[ \mathbb{E}[f(C_a^{\varphi}(\mathcal{Q}): a \in  I) \cdot g(C_a^{\varphi}(\mathcal{Q}): a \in  J)] \leq  \mathbb{E}[f(C_a^{\varphi}(\mathcal{Q}): a \in  I)] \cdot \mathbb{E}[g(C_a^{\varphi}(\mathcal{Q}): a \in  J)] \ .\]

    Note that $C_a^{\varphi}(\mathcal{Q}) = \varphi(\abs{\mathcal{Q}}_a) = \varphi(X^a)$ is a nondecreasing function of $\{ X^a_i \}_{i \in  [R]}$, since $\varphi$ is nondecreasing and $X^a = \sum_{i \in T} X^a_i$. Thus in order to show that $\{ C_a^{\varphi}(\mathcal{Q}) \}_{a \in [n]}$ are negatively associated, it suffices to show that $\{ X^a_i \}_{i \in  [R],a \in [n]}$ are negatively associated (see Proposition P$_6$ of \cite{JP83}).

    For fixed $i \in [R]$, $\{ X^a_i \}_{a \in [n]}$ are negatively associated because it corresponds to a permutation distribution of $(0,\ldots,0,1,\ldots,1)$, with $n-\frac{x_{\varphi}n}{h}$ zeros and $\frac{x_{\varphi}n}{h}$ ones, since it describes a random subset of size $\frac{x_{\varphi}n}{h}$ (see Definition 2.10 and Theorem 2.11 of \cite{JP83}). Then, using the fact that the families $\{ X^a_i \}_{a \in [n]}$ are mutually independent, we obtain that $\{ X^a_i \}_{i \in  [R],a \in [n]}$ are negatively associated (see Property P$_7$ of \cite{JP83}). Using \cite{DR98}, this establishes the claimed Chernoff-Hoeffding bound.
    
    Since there are at most $(h+1)^R$ choices of $T$ and $\mathcal{Q}$, a union bound gives:
    \[ \mathbb{P}\Big(\exists C,\mathcal{Q} : \abs{C^{\varphi}(\mathcal{Q}) -\psi^{\varphi}_{\abs{T},h} n} > \eta n\Big) \leq 2(h+1)^R \text{exp}\Big(-2\Big(\frac{\eta}{\varphi(R)}\Big)^2n\Big)\ .\]

    Thus with probability at least $9/10$, we have that $\abs{C^{\varphi}(\mathcal{Q}) -\psi^{\varphi}_{\abs{T},h} n} \leq \eta n$, since we have taken $n \geq \eta^{-2}R\varphi(R)^2\log(20(h+1))$. So there must exists some choice of $\mathcal{P}$ that satisfies the first and second constraints of partitioning systems. Thus, we can enumerate over all choices of $\mathcal{P}$ in time exp($Rn\log(n))\cdot\text{poly}(h)$ to find such a partitioning system.
\end{proof}

\section{Proof of Theorem \ref{theo:HardnessMA}}
\label{app:HardnessMA}
\begin{proof}
We show that $\varphi$-\textsc{Resource Allocation} corresponds to $\varphi$-\textsc{MaxCoverage} under a matroid constraint. Given an instance of $\varphi$-\textsc{Resource Allocation}, consider the partition matroid $\mathcal{M}$ on $[\sum_{i \in[k]} m_i] := [m_1] + \ldots + [m_k]$, where $(B_i)_{i \in [k]} := ([m_i])_{i \in [k]}$ is a partition of the ground set and the cardinality constraint  for each $i$ is to $d_i=1$. 

Here, $I \subseteq [\sum_{i \in[k]} m_i]$ is an independent set of the matroid iff $\abs{I \cap B_i} \leq d_i = 1$, for all $i \in [k]$. This corresponds to each agent $i \in [k]$ selecting at most one element from the available $m_i$ choices. In other words, we have a bijection $f$ between tuples $(A_1, \ldots, A_k) \in \mathcal{A}_1 \times \ldots \times \mathcal{A}_k$ and maximal independent sets (bases) of $\mathcal{M}$ such that $W^{\varphi}(A) = C^{\varphi}(f(A))$. Therefore, Theorem~\ref{theo:AlgoMat} leads to a polynomial-time $\alpha_{\varphi}$-approximation algorithm for $\varphi$-\textsc{Resource Allocation}. 

For the hardness part of the theorem, the proof is exactly the same as in Theorem \ref{theo:Hardness}, but instead of $\mathcal{F} := \set{F^v_{\beta}, v \in V, \beta \in [L]}$ and $k = \abs{V}$, we take  $k = \abs{V}$ to be the number of agents and $\mathcal{A}_i := \set{F^{v_i}_{\beta}, \beta \in [L]}$ where $V =\set{v_1, \ldots, v_k}$. Hence, instead of subsets of $\mathcal{F}$ of size $k$, we only consider one set $F^v_{\beta} \in \mathcal{F}$, for each $v \in V$. The function we maximize in the reduction remains unchanged. 

To establish completeness, we note that the subset described is already of the right form and, hence, the arguments continue to hold. For proving soundness, the constraint on the shape of the subset of $\mathcal{F}$ only helps us, since it gives more constraints on the given subset from which we want to construct a labeling. Therefore, both parts of the proof work and the {\rm NP}-hardness follows.
\end{proof}

\end{document}